\documentclass[11pt]{article}
\usepackage{fullpage}
\usepackage{times}
\usepackage[numbers]{natbib}
\usepackage{subfigure}
\usepackage{graphicx,amssymb,amsmath}
\usepackage{epsfig}
\usepackage{algorithm}
\usepackage{latexsym}
\usepackage{url}
\usepackage{color}
\usepackage{graphicx,amssymb,amsmath}
\usepackage{url}
\newtheorem{theorem}{Theorem}[section]
\newtheorem{lemma}[theorem]{Lemma}

\newtheorem{corollary}[theorem]{Corollary}
\newtheorem{definition}[theorem]{Definition}
\newtheorem{assumption}[theorem]{Assumption}
\newtheorem{example}[theorem]{Example}
\newcommand{\sq}{\hbox{\rlap{$\sqcap$}$\sqcup$}}
\newcommand{\qed}{\hspace*{\fill}\sq}
\newenvironment{proof}{\noindent {\bf Proof.}\ }{\qed\par\vskip 4mm\par}

\allowdisplaybreaks
\begin{document}
\setcitestyle{aysep={},authoryear}
\title{Optimal commitments in auctions with incomplete information}
\author{Zihe Wang \\
   IIIS, Tsinghua University \\
   wangzihe11@mails.tsinghua.edu.cn
   \and
   Pingzhong Tang \\
   IIIS, Tsinghua University \\
   kenshin@tsinghua.edu.cn \\
   }
\maketitle 

\begin{abstract}
We are interested in the problem of optimal commitments in rank-and-bid based auctions, a general class of auctions that include first price and all-pay auctions as special cases. Our main contribution is a novel approach to solve for optimal commitment in this class of auctions, for any continuous type distributions. Applying our approach, we are able to solve optimal commitments for first-price and all-pay auctions in {\em closed-form} for fairly general distribution settings.

The optimal commitments functions in these auctions reveal two surprisingly opposite insights: in the optimal commitment, the leader bids passively when he has a low type. We interpret this as a credible way to alleviate competition and to collude. In sharp contrast, when his type is high enough, the leader sometimes would go so far as to bid above his own value. We interpret this as a credible way to threat. Combing both insights, we show via concrete examples that the leader is indeed willing to do so to secure more utility when his type is in the middle.

Our main approach consists of a series of nontrivial innovations. In particular we put forward a concept called equal-bid function that connects both players' strategies, as well as a concept called equal-utility curve that smooths any leader strategy into a continuous and differentiable strategy.
We believe these techniques and insights are general and can be applied to similar problems.
\end{abstract}







\section{Introduction}

First price auction is one of the most well-known single-item auction formats. In first price auction, bidders simultaneously submit sealed bids to the seller, who sells the item to the highest bidder at his/her bid. The auction enjoys many desirable properties: simple, intuitive and easy to implement. Furthermore, in symmetrical settings, the auction has a unique efficient Bayes Nash equilibrium (BNE)~\cite{chawla2013auctions}. This is in contrast to second price auction that may have many inefficient equilibria (e.g., for bidders types drawn from interval [0,1], one bidder bids 1 and the others bid 0 is a BNE in second price auction).

In the meanwhile, first price auction suffers from several criticisms. For one, in complete information settings, the auction (and its extension {\em generalized first price auction}) sometimes does not have a pure Nash equilibrium and is practically observed to be unstable \cite{edelman2007strategic,edelman2007internet,borgers2013equilibrium}. Perhaps more surprisingly, in incomplete information settings where bidders have asymmetrical type distributions, its Bayes Nash equilibrium is extremely difficult to solve or characterize. In fact, this has been one of the most elusive open problems in the literature of auction analysis \cite{Vic61,lebrun1999first,fibich2011numerical,hartline2014price}. To date, the problem has only been known to have closed-form solutions in very specific settings such as two-bidder asymmetric uniform distributions
\cite{kaplan2012asymmetric,fibich2012asymmetric}.

\subsection{Commitment and related work}
In this paper, we relax the assumption of BNE by considering an alternative solution concept: commitment (aka. Stackelberg equilibrium)~\cite{Stackelberg,leader,commitEC06,extensiveEC10}. In a Stackelberg equilibrium, a leader finds an optimal strategy to commit to, given that a follower knows the leader's committed strategy in advance and best responds to it. Stackelberg equilibrium is particularly useful when one player has credibility to commit. It is well known that commitment weakly increases the leader's utility compared to Nash equilibrium. Furthermore, there are efficient algorithms to compute it in basic settings~\cite{commitEC06,extensiveEC10}. In fact, commitment, together with its use in security domains, has been regarded as one of most appealing applications of game theory over the past decade (see ~\cite{Tambe2011}).

The concept of commitment has been observed in the domain of auction design, even though sometimes implicitly. Note that early bidding and sniping in online auctions (e.g., eBay auctions) can be regarded as two forms of commitment~\cite{Roth2002,Gray2013}. An advertiser that has a ''passive'' image (i.e., rarely changes its bid, or always submits low bids) in sponsored search auctions can be seen as another form of commitment\footnote{Observations based on communication with researcher Yicen Liu from Baidu Inc.}.
Abraham et. al. study the case where there exists a super bidder that has access to more information than others and study how this will affect seller's revenue in an alternative solution concept called tremble robust equilibrium. Their setting is similar to the commitment setting but not the same~\cite{abraham2013peaches}.
Skreta considers another type of commitment where the auctioneer is lack of credibility to reserve the item and studies how this lack of commitment affects revenue~\cite{Skreta}. Our work draws close connections between commitments and collusions in certain auctions.

A closely-related parallel work is the one by Xu and Ligett~\cite{xuyunjian2014}. They characterize optimal commitment for first price auction with complete information. For the incomplete information case, they assume that the bidders' types are drawn from discrete distributions and proof a partial property that the commitment function can be divided into pieces. In comparison, we consider continuous distribution settings and obtain strong/closed-form characterizations for a more general class of auctions. Our technical approach is also different and yields new insights.

\subsection{Our contributions}
In this paper, we study the optimal commitment in {\em rank-and-bid based auctions} with incomplete information (Bayesian setting).
Rank-and-bid based auction~\cite{chawla2013auctions} is a general class of auctions in which bidder's payment is decided by its own rank and bid, including first price and all-pay auctions as special cases. Our main contribution is a general approach to solve and characterize optimal commitment in this class of auctions, for any continuous type distributions. In particular, applying our approach, we are able to solve optimal commitments for first-price and all-pay auctions in {\em closed-form} for fairly general distributions settings. Our approach and results on these auctions in a sense mitigate the difficulties of deriving a game-theoretical prediction in first price and all-pay auctions with asymmetric type distributions. Our approach consists of several nontrivial techniques. We dedicate Section~\ref{sec:app} to introduce the technical contribution (our main contribution). Here, we focus on the economic interpretations of our results.

Our results on first price and all-pay auctions reveal certain striking findings: the leader bids very passively when his type is low. Even worse, he bids 0 when his type is below a threshold (depending on type distributions). This is against the common sense in first price and all pay auctions that bidding 0 has no chance of winning at all. However, a close scrutiny enlightens us otherwise: by committing to a passive image, the leader (credibly) ensures the follower that he has no intention to compete when he has a low type, thus effectively brings down the follower's bid, since the follower does not know the leader's actual type and views the leader as a mixed strategy. As a result, the leader eventually wins the auction with less competition when he has a high type. We also note that such passive bidding behaviors had been observed in major search engines such as Yahoo\footnote{\url{http://webscope.sandbox.yahoo.com/catalog.php?datatype=a}} and Baidu (before they switch to GSP).

Furthermore, the commitment solutions are largely consistent with the collusive behaviors studied in first-price auction~\cite{mcafee1992bidding,marshall2007bidder,aryal2013testing,lopomo2011bidder,pesendorfer2000study}: players coordinates to bring down the prices. Our results further suggest that such collusive behaviors are stable: the trust between the players is built on the rationality of the follower, as well as the leader's credibility to commit.

However, it is important to realize that the leader is not always passive and collusive. We also observe (Corollary~\ref{corollary1}) that the leader sometimes overbids his own valuation! Our insight here is that by placing aggressive bids on high values, the leader creates threat that effectively drives the follower's bid to zero, which guarantees sufficiently high utility for the leader when his type is in the middle.

Finally, note that, even though the leader throws the game with relatively large probability ($78\%$ in one of our examples) when his type is small, his valuation in this case is small as well, as a result, the utility loss is insignificant. In the meanwhile, when his type is high, even the leader suffers from a deficit by overbidding, he does so with small probability ($4\%$ in the same example), the loss in this case is insignificant as well. The losses on both ends are compensated by the gain in the middle.

\section{The settings}

We consider a single item auction with two bidders, one called the {\em leader} $A$ (male) and the other called the {\em follower} $B$ (female). Bidder $A$ has a private valuation $x$ drawn from distribution $F_1$ with support $[a_1, a_2]$, while bidder $B$'s private valuation $y$ is drawn from distribution $F_2$ with support on $[b_1, b_2]$. We use $f_1$ and $f_2$ to denote the probability density of $F_1$ and $F_2$ respectively. We also sometimes write $A=x$ to denote the case where $A$'s type is $x$. Similar for $B=y$.

Leader $A$ commits to a Bayesian strategy $s_A: [a_1,a_2]\rightarrow \triangle R$, where $\triangle R$ denotes the set of bid distributions on $R$. He announces this strategy and the follower $B$ best responds to the leader's committed strategy via a single bid\footnote{It is easy to see that it is never beneficial to use a mixed strategy for $B$.}. We follow the standard definition~\cite{commitEC06} of Bayesian commitment that the leader only announces his strategy, i.e., the function $s_A(\cdot)$, without revealing his actual type. Being able to commit increases utility for the leader compared to the utility in BNE. To make this statement concrete in our setting, consider the following example\footnote{We will use first price auction as a running example throughout the paper}.

\begin{example}
Both players' types are drawn uniformly from $[0,1]$. Let $s_A(v)=v^2/2$. Clearly, the follower $B$ must never bid more than $0.5$ in this case. In fact, her utility when bidding $t\leq 0.5$ is $(y-t)\sqrt{2t}$. $B$'s best strategy is $s_B(y)=y/3$. The expected utility of $A$ is:
$$\int_{0}^{\sqrt{2/3}}(x-\frac12x^2)\cdot \frac32x^2dx + \int_{\sqrt{2/3}}^1(x-\frac12x^2)dx=0.2029$$
The first term consider the case where $x$ is in $[0,\sqrt{2/3}]$, while the second term considers the case where $x$ is in $[\sqrt{2/3},1]$. Compared to the unique symmetric BNE where each bidder bids half of the value, $A$'s expected revenue is $1/6=0.167$. In other words, committing to $s_A$ increases A's utility by $21\%$.
Theorem~\ref{theorem1st} says committing to the optimal strategy generates A's utility $0.22$, increasing by $32\%$.
\end{example}

As mentioned, our goal is to solve for the optimal $s_A$ in a general class of auctions called the {\em rank-and-bid based auctions}. It is worth mentioning that computing optimal commitment in general Bayesian game is NP-hard~\cite{commitEC06}.

\begin{definition}
{\em Rank-and-bid based auctions}
\begin{itemize}
\item Allocation rule: the item is always allocated to the highest bidder.
\item Payment rule: A bidder's payment depends only on its own bid and whether it wins or not.
The payment is $p(b)=p^p(b)$ if the bidder loses with bid $b$.
The payment is $p(b)=p^p(b)+p^w(b)$ if bidder wins with bid $b$.
Here $p^p$ and $p^w$ are differentiable functions representing the agent's payments for participation and winning respectively.
\end{itemize}
\end{definition}

Clearly, both first price and all-pay auctions belong to this class. For first price auction, $p^p(t)=0$ and $p^w(t)=t$; while for all-pay auction, $p^p(t)=t$ and $p^w(t)=0$.

We make a mild assumption that $p^p(b)+p^w(b)$ is strictly increasing, $p^p(0)=p^w(0)=0$, and $p^p,\ p^w$ are differentiable functions. Again, the two auctions mentioned above satisfy the assumption. In this paper we define  $p^w,p^p:R\rightarrow R$. In fact, it does not matter when $p^p$ and $p^w$ take inputs on $R^-$ because $A$ never submits any negative bid. We do so only for ease of exposition.

%


\begin{assumption}
\label{ass1}
When $B$ has multiple best responses, she will choose the one that maximizes $A$'s winning probability, i.e., she will submit the lowest best response.
\end{assumption}

\begin{assumption}
\label{ass2}
When there is a tie, the good will be assigned to the $B$.
\end{assumption}


In fact, our main results do not depend on the assumptions above, as we show formally in the appendix that neither of the assumptions is necessary.




\begin{definition}$P_A[t,s_A]$ denotes $A$'s winning probability when he is at type $t$, bidding $s_A(t)$, while $P_B[t,s_A]$ is $B$'s winning probability when bidding $t$ against $A$'s strategy $s_A$.
In circumstances where there is no ambiguity, we use $P_A[t],P_B[t]$ instead. \label{def2}
\end{definition}

\begin{definition}\label{def3} We use $u_A$ and $u_B$ to denote $A$'s and $B$'s expected utility respectively.
\begin{displaymath}
u_A(s_A) = E_{x\sim F_1} (x\cdot P_A[x,s_A]-p^p(s_A(t))-p^w(s_A(t))\cdot P_A[x,s_A])
\end{displaymath}
while $u_B$ can be conveniently expressed as,
\begin{displaymath}
u_B(y)=\sup_{t\geq 0}(y\cdot P_B[t]-p^p(t)-p^w(t)\cdot P_B[t])
\end{displaymath}
i.e. $u_B(y)$ is the largest utility $B$ can achieve with value $y$.
Though $B$'s valuation is on $[b_1,b_2]$, we extend its definition to $[0, b_2]$, and $F_2[x]=0, x< b_1$.
\end{definition}

Note that we use $\sup$ rather than $\max$ here because we still need to show that $\max$ is always attainable. We do so in Lemma~\ref{lemma0}. We also note that bidding $0$ yields a nonnegative utility for $B$, so $u_B\geq 0$ and $u_B(0)=0$.

\begin{definition}
\label{def4}
Let $S_B(y,s_A)$ denote $B=y$'s best responses against $A$'s strategy $s_A$. In circumstance where there is no ambiguity, we use
 $S_B(y)$ instead.
\end{definition}

We will prove this definition is well-defined in Lemma~\ref{lemma0}, i.e., $S_B(y)$ is nonempty.

\section{Approach Sketch}
\label{sec:app}

As mentioned, the main contribution of this paper is to put forward a general approach for optimal commitments in rank-and-bid based auctions. The approach can be sketched by a few insights and steps. It is useful to understand these insights and steps before we proceed to the details.

\subsection{Difficulties of the problem}

The main difficulty of the problem is that the follower's best response cannot be represented as a closed-form function of the leader's Bayesian strategy (also a function). This difficulty further prevents us to obtain a closed-form representation of the leader's winning probability and utility. As a result, standard functional optimization techniques do not apply. To appreciate the difficulty, it might be helpful to compare to the problem of finding Bayes Nash equilibrium with two asymmetric bidders in first-price auction --- one of the most elusive open problems in the analysis of auctions \cite{kaplan2012asymmetric,hartline2014price}.
The main barrier in that literature is exactly the difficulty to represent one's best response as a concise function of the other's strategy.

Furthermore, the leader's optimal strategy is not necessarily continuous and differentiable, making it difficult to optimize and analyze. In contrast, consider again the literature of Bayes Nash equilibrium in first price auction: standard techniques exist to guarantee the equilibrium strategy to be monotone and differentiable.

Previous work on this problem focuses on different cases of finding optimal commitment in first price auction with complete information
\cite{xuyunjian2014}. They also discuss the Bayesian case where the follower's type is drawn from a discrete distribution. They obtain a partial property of the optimal commitment in the Bayesian case. Their approach and result do not extend to our continuous case.

\subsection{Key innovation: representation via equal-bid function}

A key insight here is to represent everything (the strategies, utilities and winning probabilities of both players) as functions of $g$, coined the {\em equal-bid function}, that maps a leader's type to a follower's type. Intuitively, $g(t)$ is the follower's type at which she submits {\em the same} bid as the leader does when the leader is at type $t$. In other words, $g(t)$, later proved to be monotone, can be seen as a cutoff type between win and lose for the follower. As one can imagine, together with the cumulative distribution function of the follower's type, we can represent the leader's winning probability (hence utility) as a function of $g$. A similar but different idea has appeared in \cite{lebrun1999first} where they use {\em inverse bid function} to represent the best response of one another.

\subsection{Step one: to sort the leader's strategy}

Even with the help of the equal-bid function, we still cannot compute the optimal leader's strategy. The next obstacle is the optimal leader strategy might not be monotone or differentiable. So, our first effort is to sort the leader's strategy. We prove that, for any leader's strategy (optimal or not), one can sort it into a monotone function that preserves the follower's best response without hurting the leader's utility. In other words, the leader always prefers the strategy after sorting.

\subsection{Step two: to smooth the leader's strategy and the equal-utility curve}

The more difficult part is to smooth the leader's strategy, i.e., to turn it into a continuous and differentiable function. To achieve this goal, we introduce the second innovation called {\em equal-utility curve}. Roughly, fix the follower's type, the equal-utility curve represents a leader's strategy such that the follower gets the same utility no matter what she bids. We first show that such curve always exists. Furthermore, the supremum (over all follower's types) of all such curves defines a new leader strategy that enjoys the following important properties: it is continuous, left and right differentiable, preserves the follower's best response, and weakly improves the leader's utility. Till now, one can truly focus on monotone, continuously differentiable leader strategies.

\subsection{Step three: bijection between the leader strategies and equal-bid functions}

It turns out that we can find a bijection between the set of monotone, continuously differentiable leader strategies and the set of continuous and monotone equal-bid functions. As a result, we can restrict attentions to optimize over equal-bid functions and the standard Lagrangian method applies.

 We conclude with a characterization of the optimal commitment for general follower type distributions and compute the closed-form optimal commitments for both all-pay and first-price auctions when the follower has distribution uniform type distributions.

\section{Sort and Smooth the leader's strategy}

We first transform an arbitrary leader strategy into a continuous weakly increasing strategy and
show that this transformation does not hurt the leader's expected utility.

There are some additional properties (e.g., differentiability) for the transformed strategy, which we prove in the next section


We first show that the notion of ``best response'' is well defined for the follower $B$.

\begin{lemma}
\label{lemma0}
For any B's valuation $y$, $B$ has a best response, i.e. $u_B(y)$ can be attained by some bid and the smallest bid exists among the best responses.
\end{lemma}

By Assumption~\ref{ass2}, $B$ always chooses the smallest bid among all best responses.

\subsection{Sort $s_A$}

For arbitrary strategy $s_A$, the support of $s_A(v)$ on value $v$ may not be a single bid.
Function $s_A$ could also be nonmonotone.
The following lemma says, to find the optimal commitment, it suffices to consider the strategies with the desirable properties below.
\begin{lemma}For any strategy $s_A$ for A, we can sort it into a new strategy function $\breve{s}_A$ such that (1) $\breve{s}_A(v)$ is a deterministic bid for any $v$ and is weakly increasing in $v$
(2)the best response of $\breve{s}_A$ remains the same as $s_A$
(3) $\breve{s}_A$ yields at least the same utility for the leader as $s_A$.
\label{lemma1}
\end{lemma}

From now on, for ease to presentation, we use $s$ to denote $\breve{s}_A$. So $s$ is an weakly increasing, nonnegative strategy function.




The following example shows how to calculate $u_B$ in first price auction.
\begin{example}Both bidders' value distribution are uniform on [0,1].
\begin{displaymath}
s_A(x) = \left\{ \begin{array}{ll}
x/4 & x\leq 0.4\\
x-0.3 & x>0.4
\end{array} \right.
\end{displaymath}
By definition $u_B(y) = \max \{\max_{t\leq 0.1}(y-t)\cdot4t, \max_{t>0.1}(y-t)(t+0.3)\}$, we have
\begin{displaymath}
u_B(y) = \left\{ \begin{array}{ll}
y^2 & y\in[0, 0.2)\\
0.4y-0.04 & y\in[0.2,0.5]\\
(y+0.3)^2/4 & y\in(0.5,1]
\end{array} \right.
\end{displaymath}
 \label{eg1}
\end{example}

The following theorem says the set of best responses of $B$ at type $y_1$ generally does not intersect with the set of best responses at type $y_2$ and these sets are well sorted by the value of $y$, except for the following special case.
\begin{theorem}For any $B$'s valuations
$y_1<y_2$, if $\exists a\in S_B(y_1), b\in S_B(y_2)$ but $a> b$, then
$S_B(y_1)\subseteq S_B(y_2)$, $u_B(y_1)=u_B(y_2)=0$ and $P_B[b]=P_B[a]=0$.
\label{lemma2}
\end{theorem}

We can now prove that the follower's utility is continuous an monotone.

\begin{lemma}
$u_B(y)$ is continuous and weakly increasing.\label{lemma3}
\end{lemma}

\subsection{Smooth $s_A$}

To smooth $s_A$ into a continuous and differentiable function, we now introduce an important innovation of our approach: the equal-utility curve.

\begin{definition}
\label{def5}Define equal-utility curve $eu_B(\cdot,\cdot)$: $[0,b_2] \times (a_1,a_2] \rightarrow \mathbb{R}$,
such that $eu_B(y,x)$ is the only solution (solve for $t$) of
\begin{eqnarray}
u_B(y)=F_1[x]y-p^w(t)F_1[x]-p^p(t) \ x\in(a_1,a_2]\label{eqa5}
\end{eqnarray}
\end{definition}
The interpretation of $eu_B(y,\cdot)$ is that, any value of $eu_B(y,\cdot)$ (as a function of $x$) is a best response of the follower at type $y$.

Consider again Example~\ref{eg1}, where in first price auction, $F_1[x]=F_2[x]=x \ \forall x\in[0,1]$. In first price auction, the definition of $eu_B$ above is simplified as
$$eu_B(y,x)=y-\frac{u_B(y)}{x}$$
When $B$'s value $y=0.5$, her best response is to bid $0.1$ with utility $0.16$.
So $u_B(0.5)=0.16$, and EU curve is $eu_B(0.5,x)=0.5-\frac{0.16}{x}$, shown in Fig~\ref{fig:1}.
If $A$ uses strategy
\begin{displaymath}
\max\{eu_B(0.5,x),0\} = \left\{ \begin{array}{ll}
0 & x\in[0, 0.32)\\
0.5-\frac{0.16}{x} & x\in[0.32,1]
\end{array} \right.
\end{displaymath}
then the utility of $B$ when $y=0.5$ is the same for any bid in $[0,0.34]$.

\begin{figure}[ht]
  \centering
  \includegraphics[width=4cm]{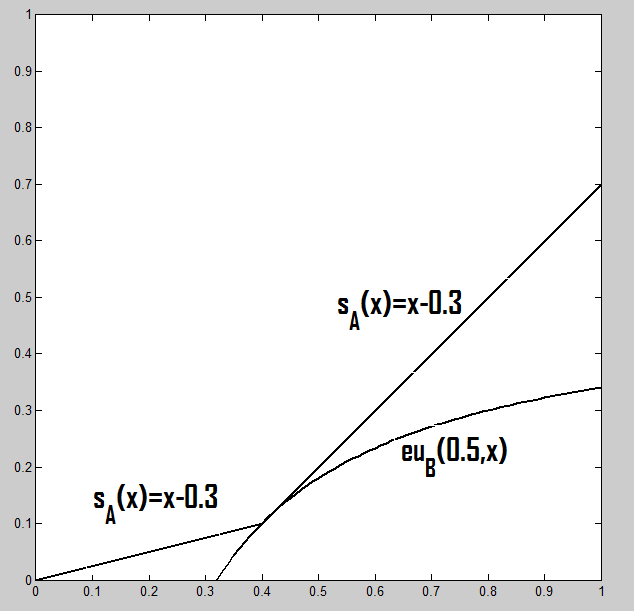}\\
  \caption{Equal-Utility(EU) curve. When $y=0.5$, the utility of bidding $0$ or $0.3$ are always $0.16$.}
  \label{fig:1}
\end{figure}

First of all, function $eu$ is well defined. Given $x,y$, when $t$ increase, the right hand side of Equation~(\ref{eqa5}) strictly decreases( $F_1[x]$ is positive). So there exists a unique solution (maybe negative). Function $eu$ represents the leader's strategy against which the follower will achieve the same largest utility no matter what the follower bids.
It's easy to check that $eu_B(0,\cdot)=0$.

The following lemma is a technical gadget,
which prepares us to prove that $eu$ function is weakly increasing and differentiable.
It's also used in the proof of Theorem~\ref{theorem8}.
\begin{lemma}
Given $a>0$, $(p^w)'$ and $(p^p)'$ are nonnegative, $(p^w)'+(p^p)'>0$. If $t(a,b)$ satisfies
\begin{eqnarray}
b+ap^w(t)+p^p(t)=0\nonumber
\end{eqnarray}
\label{lemma3.25}
we can prove $t(a,b)$ exists, $\frac{\partial t}{\partial a}(a,b)$ and $\frac{\partial t}{\partial b}(a,b)$ are continuous. In particular
\begin{displaymath}
\frac{\partial t}{\partial a}(a,b)=\frac{-p^w(t)}{a(p^w)'(t)+(p^p)'(t)} \qquad \frac{\partial t}{\partial b}(a,b)=\frac{-1}{a(p^w)'(t)+(p^p)'(t)}
\end{displaymath}
\end{lemma}

\begin{lemma}
(1)$eu_B(y,x)$ is weakly increasing and differentiable in $x$.
$eu_B(y,x)$ is continuous in $y$.
(2)$eu_B(y,x)\leq s_A(x)$
\label{lemma3.5}
\end{lemma}

We are now ready to introduce the smooth method:create the $s^*_A$ using $eu_B$.
\begin{definition}$s^*_A(x)=\sup_{y\in[0,b_2]}eu_B(y,x), \forall x\in(a_1,a_2]$\label{def6}
\end{definition}

We will prove that strategy $s^*_A$ yields at least the same revenue for $A$ as strategy $s$. The outline is following.
First, we prove that, though the leader's bid distribution changes, we keep the utility of the follower remain the same.
Second, we prove the best response of $B$ is still the best response after smoothing, for any follower's value (Lemma~\ref{lemma6}).
Third, we prove the leader's winning probability does not change.
At last, combined the fact that leader's bid is weakly decreasing, we prove that leader's utility weakly increases after smoothing.(Theorem~\ref{lemma7})

An mathematical view of the motivation of $s^*_A(x)$ can be found in Remark 2.
The idea is to suppress the bids of the leader while maintain his winning probability.

Consider again Example~\ref{eg1},
\begin{displaymath}
s^*_A(x) = \left\{ \begin{array}{ll}
x/4 & y\in[0, 0.4)\\
x-0.3 & y\in[0.4,0.65]\\
1-\frac{0.65^2}{x} & x\in(0.65,1]
\end{array} \right.
\end{displaymath}
 \label{eg2}

We now prove some basic properties of $s_A^*$ that will be used later.
\begin{lemma}
\label{lemma4}
(1) For any $x\in(a_1,a_2]$, $(x,s_A^*(x))$ must lie on some EU curve.

(2) When $s^*_A(x)>\lim_{t\rightarrow a_1}s^*_A(t)$, $s_A^*(x)$ strictly increases.

(3) When $s^*_A(x)=\lim_{t\rightarrow a_1}s^*_A(t)$, $(x,s^*_A(x))$ lies on $eu_B(p^w(s^*_A(x)),\cdot)$, and $u_B(p^w(s^*_A(x)))=0$, $p^p(p^w(s^*_A(x)))=0$.

(4) For any $x$, we have $s^*_A(x)\leq s_A(x)$.

(5) $s^*_A(x)$ is continuous.\footnote{The limitation of continuous may be non continuous, so this argument is not trivial. Consider $y_k(x)=kx, x\in[0,1/k]$ and constant zero for $x\leq 0$ and constant one for $x\geq 1/k$. $y_k$ is continuous but $sup_ky_k$ is not.}
\end{lemma}

For simplicity, we use $P^*_B[t]$ and $S^*_B(y)$ instead of $P_B[t,s^*_A]$ and $S_B(y,s^*_A)$.
Next, we study how the follower's utility and best response would change in $s$ and $s^*_A$.
\begin{lemma}
\label{lemma6}
(1) When $A$'s strategy $s$ is changed to $s^*_A$, $u_B(y)$ remains the same.

(2) If $t\in S_B(y)$ then $t\in S^*_B(y)$,
$S_B(y)\subseteq S^*_B(y)$.
If $P^*_B[t]\neq P_B[t]$ then $u_B(y)= 0$ and $t=\lim_{x\rightarrow a_1}s^*_A(x)$
\end{lemma}

The lemma below draws connections between equal-utility curve and $s^*_A$.
\begin{lemma}
\label{lemma5}
(1) If $(x_0,s^*_A(x_0))$ lies in eu line $eu_B(y_0,\cdot)$, then $s^*_A(x_0)\in S^*_B(y_0)$.
(2) If $s^*_A(x_0)\in S^*_B(y_0)$ and $s^*_A(x_0)\neq \lim_{x\rightarrow a_1}s^*_A(x)$ then $(x_0,s^*_A(x_0))$ lies in eu line $eu_B(y_0,\cdot)$.
\end{lemma}

The best response set $S^*_B$ of the follower in $s^*_A$ is an a superset of the $S_B$.
Remind that the best response of the follower is sorted.
So, $S^*_B(y)$ is bounded by any element in $S^*_B(y-\epsilon)$, and $S^*_B(y+\epsilon)$.
Thus, it's bounded by $S_B(y-\epsilon)$, and $S_B(y+\epsilon)$.
We then prove that for most of the types, the winning probability of the leader will not change.

If $\exists x$ such that $s^*_A(x)=\lim_{t\rightarrow a_1}s^*_A(t)$, define
\begin{definition}
$\hat{x}=\sup\{x|\lim_{t\rightarrow a_1}s^*_A(t)=s^*_A(x)\}$.
\end{definition}
Since $s^*_A$ is continuous, $s^*_A(\hat{x})=\lim_{t\rightarrow a_1}s^*_A(t)$.
The following lemma tells what $s$ is if $s^*_A$ has a constant interval in the beginning.
\begin{lemma}
If $\exists x$ such that $s^*_A(x)=\lim_{t\rightarrow a_1}s^*_A(t)$, we have
\label{sconstant}
$s^*_A(x)=s_A(\hat{x})\ \forall x<\hat{x}$.
\end{lemma}
Combining with the fact that the bids decrease in $s^*_A$ (thus lower payment)  and still the same winning probability, we prove that the expected utility does not decrease.
\begin{theorem}
\label{lemma7}
By using $s^*_A$ instead of $s$, the expected utility of $A$ does not decrease.
\end{theorem}

\section{Bijective mapping between $s^*_A$ and $g$}

The final step is to show that every $s^*_A$ can be represented by a function $g$ and we shall restrict attention to optimize such $g$ instead.
\begin{figure}[h]
\label{fig:2}
  \centering
  \includegraphics[width=8cm]{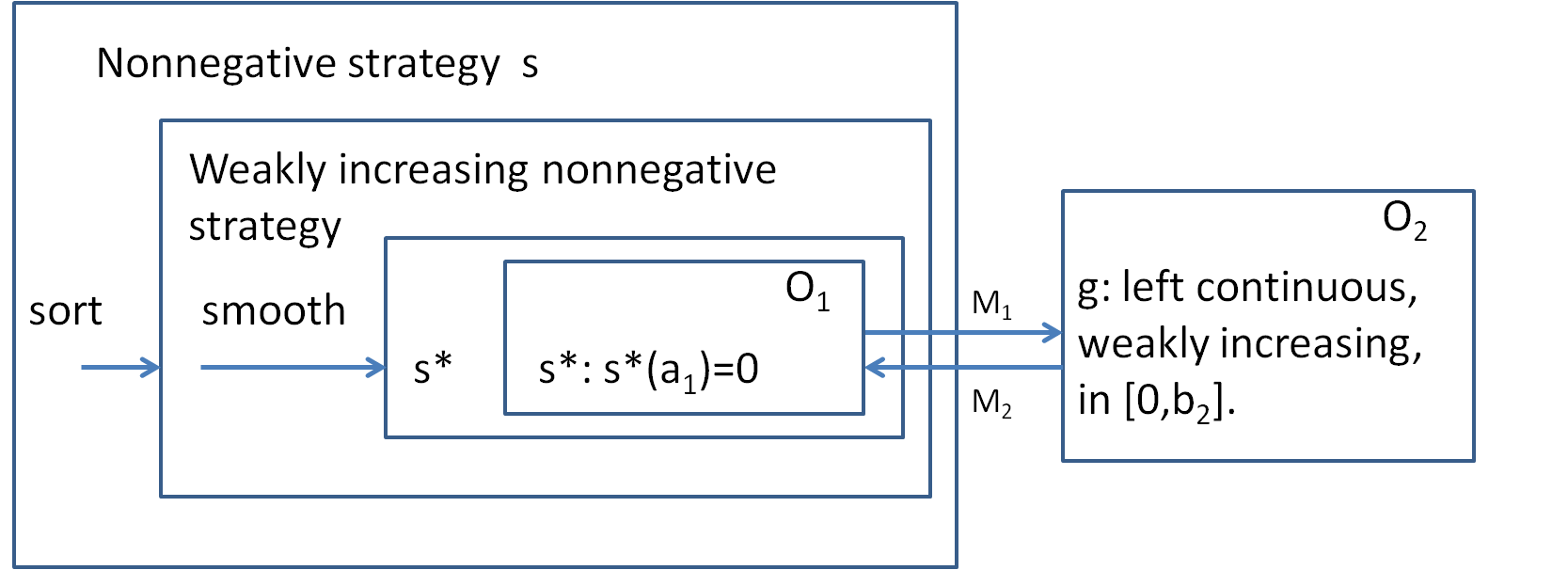}\\
  \caption{$M_1$ is bijective between $O_1$ and $O_2$}
\end{figure}
\begin{definition}
$\hat{y}=\sup\{y|u_B(y)=0\}$
\end{definition}
\begin{definition}
$\forall x>a_1,\ Y(x)=\{y| eu_B(y,\cdot) \textrm{passes through point } (x,s^*_A(x))\}$.
\end{definition}
\begin{lemma}(1) $Y(x)$ is closed.
(2)$Y(x)\geq \hat{y}$.
(3)For all $x_1<x_2$, $Y(x_1)\leq Y(x_2)$.
(4)$(\hat{y},b_2]\subseteq \cup_x Y(x)$
(5)$Y(x)$ is an interval or is a unique number. $Y(x)$ contains only one element for almost all $x$.
\label{lemma8}
\end{lemma}

\begin{definition}
\label{def7} Equal-bid function $g(x)=\min Y(x), \forall x\in(a_1,a_2]$
\end{definition}
The intuitive explanation of $g$ is: when the follower's type $y=g(x)$, one of her best response is to equal the bid $s^*_A(x)$ by the leader.
In most of the time except for countable types, if follower gives the same bid as the leader,
then the lowest type of the follower must be $g(x)$.
In other words, the follower with value $g(x)$ gives the same bid as the leader with value $x$.

{\em With equal-bid function, the winning probability of the leader has a surprisingly concise form, as shown in Lemma~\ref{lemma8.5}.}

Consider again Example~\ref{eg1}, we have
\begin{displaymath}
g(x) = \left\{ \begin{array}{ll}
x/2 & y\in[0, 0.4)\\
2x-0.3 & x\in[0.4,0.65]\\
1 & y\in(0.65,1]
\end{array} \right.
\end{displaymath}
It's easy to check that $s_A(0.45)=0.15$, $g(0.45)=0.6$, $S_B(0.6)=\{0.15\}$.
The leader with value $0.45$ bids $0.15$, same as the follower with value $g(0.45)=0.6$.

By lemma~\ref{lemma8}, $\{y| eu_B(y,x)=s^*_A(x)\}$ is closed, so $g(x)$ is well defined.
When $y<g(x)$, the leader $A=x$ beats the follower $y$ by Theorem~\ref{lemma2}.
When $y\geq g(x)$, the follower $y$ beats the leader $A=x$ by tie-breaking rule.
So we can calculate the winning probability of the leader $A=x$ using $g(x)$.

\begin{lemma}Using strategy $s^*_A$, the winning probability of $A$ with type $x$ is $F_2[g(x)]$.
\label{lemma8.5}
\end{lemma}

\begin{lemma}
(1)$g(x)$ weakly increases.
(2)$g$ is left continuous.
\label{lemma12}
\end{lemma}

Note, $s^*_A(x)$ is not yet defined on $a_1$.
In fact, bid with zero probability does not affect overall utility. We can define $s^*_A(a_1)=\lim_{t\rightarrow a_1}s^*_A(t)$ for convenience.

Now we prove the last a few desirable properties of $s^*_A$: $s^*_A$ is differentiable on both sides.
Based on the derivatives, we determine the connection between $g$ and $s^*_A$.
\begin{theorem}\label{lemma9}Assume both $(p^w)'$ and $(p^p)'$ are continuous, then
\begin{enumerate}
\item $s^*_A(x)$ is left-hand differentiable and right-hand differentiable.
\item $s^*_A$ can be solved from $\int^{x}_{a_1}f_1(t)g(t)dt=p^w(s^*_A(x))F_1[x]+p^p(s^*_A(x))-p^p(s^*_A(a_1))$.
\end{enumerate}
\end{theorem}

By (2), we obtain that in first price auction $s^*_A(x)=\frac{1}{F_1[x]}\int_{a_1}^xf_1(t)g(t)dt+s^*_A(a_1)$.

In all-pay auction $s^*_A(x)=\int_{a_1}^xf_1(t)g(t)dt+s^*_A(a_1)$.

\begin{proof}
(1)
For any $x_0$, we prove $s^*_A(x)$ is right-hand differentiable at point $x_0$.
The proof for the left-hand case is similar.

If $x_2>x_0$ is close enough to $x_0$, by Lemma~\ref{lemma8},
we know $\underline{Y}(x_2)$ is close enough to $\overline{Y}(x_0)$.
Denote $y_0=\overline{Y}(x_0)=g(x_0)$ and $y_2=\underline{Y}(x_2)$
($\underline{Y}=\inf Y,\overline{Y}=\sup Y$).

Let $Q(a,b)$ denote the solution $t$ of
\begin{eqnarray}
b+ap^w(t)+p^p(t)=0\nonumber
\end{eqnarray}
By Lemma~\ref{lemma3.25}, $Q(a,b)$ exists and
\begin{eqnarray}
\frac{\partial Q}{\partial a}(a,b)=\frac{-p^w(Q)}{a(p^w)'(Q)+(p^p)'(Q)}&&\frac{\partial Q}{\partial b}(a,b)=\frac{-1}{a(p^w)'(Q)+(p^p)'(Q)}\nonumber
\end{eqnarray}
We have $\frac{\partial Q}{\partial a}$ and $\frac{\partial Q}{\partial b}$ are continuous.
We should also keep in mind that $\frac{\partial Q}{\partial a}$ and $\frac{\partial Q}{\partial b}$ are negative.

Since $(x_0,s^*_A(x_0))$ lies on $eu_B(y_0,\cdot)$
and $(x_2,s^*_A(x_2))$ lies on $eu_B(y_2,\cdot)$,
we have
\begin{eqnarray}
s^*_A(x_0)=eu_B(y_0,x_0)&& s^*_A(x_0)\geq eu_B(y_2,x_0)\nonumber\\
s^*_A(x_2)=eu_B(y_2,x_2)&& s^*_A(x_2)\geq eu_B(y_0,x_2)\nonumber
\end{eqnarray}
Rewrite these equations in terms of $Q$ we have
\begin{eqnarray}
s^*_A(x_0)&=&Q(F_1[x_0],u_B(y_0)-y_0F_1[x_0])\label{aeq1}\\
s^*_A(x_0)&\geq&Q(F_1[x_0],u_B(y_2)-y_2F_1[x_0])\label{aeq2}\\
s^*_A(x_2)&=&Q(F_1[x_2],u_B(y_2)-y_2F_1[x_2])\label{aeq3}\\
s^*_A(x_2)&\geq&Q(F_1[x_2],u_B(y_0)-y_0F_1[x_x])\label{aeq4}
\end{eqnarray}
Equation (\ref{aeq3})-(\ref{aeq2}) and divide it by $x_2-x_0$, we get
\begin{eqnarray}
\frac{s^*_A(x_2)-s^*_A(x_0)}{x_2-x_0}&\leq&\frac{Q(F_1[x_2],u_B(y_2)-y_2F_1[x_2])-Q(F_1[x_0],u_B(y_2)-y_2F_1[x_0])}{x_2-x_0}\label{aeq5}
\end{eqnarray}
Since $u$ and $F_1$ are continuous, we have
$\lim_{x_2\rightarrow x_0}u_B(y_2)-y_2F_1[x_2] = \lim_{x_2\rightarrow x_0}u_B(y_2)-y_2F_1[x_0]=u_B(y_0)-y_0F_1[x_0]$.

Since $\frac{\partial Q}{\partial a}$ and $\frac{\partial Q}{\partial b}$ are continuous,
for any $\epsilon>0$,
we can pick $x_2$ small enough such that
when $(a,b)\in [F_1[x_0],F_1[x_2]]\times[u_B(y_2)-y_2F_1[x_2],u_B(y_2)-y_2F_1[x_0]]$
we have
\begin{eqnarray}
(\frac{\partial Q}{\partial a}(a,b),\frac{\partial Q}{\partial b}(a,b))&\in &
[\frac{\partial Q}{\partial a}(F_1[x_0],u_B(y_0)-y_0F_1[x_0])\cdot(1\pm \epsilon)]\nonumber\\
&&\times
[\frac{\partial Q}{\partial b}(F_1[x_0],u_B(y_0)-y_0F_1[x_0])\cdot(1\pm \epsilon)]\nonumber
\end{eqnarray}

We use integration of the derivatives to represent the numerator of Equation (\ref{aeq5}).
Since the derivatives are very close, we can have a good upper bound.
\begin{eqnarray}
&&Q(F_1[x_2],u_B(y_2)-y_2F_1[x_2])-Q(F_1[x_0],u_B(y_2)-y_2F_1[x_0])\nonumber\\
&=&\int^{F_1[x_2]}_{F_1[x_0]}\frac{\partial Q}{\partial a}(t,u_B(y_2)-y_2F_1[x_2])dt
-\int^{u_B(y_2)-y_2F_1[x_0]}_{u_B(y_2)-y_2F_1[x_2]]}\frac{\partial Q}{\partial b}(F_1[x_0],t)dt\nonumber\\
&\leq&(F_1[x_2]-F_1[x_0])(1-\epsilon)\frac{\partial Q}{\partial a}(F_1[x_0],u_B(y_0)-y_0F_1[x_0])\nonumber\\
&&-y_2(F_1[x_2]-F_1[x_0])(1+\epsilon)\frac{\partial Q}{\partial b}(F_1[x_0],u_B(y_0)-y_0F_1[x_0])\nonumber
\end{eqnarray}
The inequality is because $\frac{\partial Q}{\partial a}$ and $\frac{\partial Q}{\partial b}$ are negative.
Then the limitation of Equation (\ref{aeq5}) would be
\begin{eqnarray}
\lim\frac{s^*_A(x_2)-s^*_A(x_0)}{x_2-x_0}&\leq& f_1(x_0)(1-\epsilon)\frac{\partial Q}{\partial a}(F_1[x_0],u_B(y_0)-y_0F_1[x_0])\nonumber\\
&&-y_2f_1(x_0)(1+\epsilon)\frac{\partial Q}{\partial b}(F_1[x_0],u_B(y_0)-y_0F_1[x_0])\nonumber
\end{eqnarray}
Since it works for any $\epsilon$, we have
\begin{eqnarray}
\lim\frac{s^*_A(x_2)-s^*_A(x_0)}{x_2-x_0}
&\leq& f_1(x_0)\frac{\partial Q}{\partial a}(F_1[x_0],u_B(y_0)-y_0F_1[x_0])\nonumber\\
&&-y_2f_1(x_0)\frac{\partial Q}{\partial b}(F_1[x_0],u_B(y_0)-y_0F_1[x_0])\nonumber\\
&=&\frac{-f_1(x_0)p^w(s^*_A(x_0))+f_1(x_0)g(x_0)}{F_1[x_0](p^w)'(s^*_A(x_0))+(p^p)'(s^*_A(x_0))}\nonumber\\
&=&\frac{\partial eu}{\partial x}(\overline{Y}(x),x)\nonumber
\end{eqnarray}
Similarly, if we use Equation (\ref{aeq4})-(\ref{aeq1}), we can get a lower bound same as the upper bound.
So $s^*_A(x)$ is right-hand differentiable at point $x_0$, and
\begin{eqnarray}
(s^*_A)'(x)=\frac{\partial eu}{\partial x}(\overline{Y}(x),x)\nonumber
\end{eqnarray}

(2) Following above, we have
\begin{eqnarray}
\partial_+s^*_A(x)&=&\frac{f_1(x)\overline{Y}(x)-p^w(s^*_A(x))f_1(x)}{(p^w)'(s^*_A(x))F_1[x]+(p^p)'(s^*_A(x))}\nonumber\\
f_1(x)\overline{Y}(x)&=&\partial_+s^*_A(x)\cdot (p^w)'(s^*_A(x))F_1[x]+p^w(s^*_A(x))f_1(x)\nonumber\\
&&+\partial_+s^*_A(x)\cdot(p^p)'(s^*_A(x))\nonumber\\
f_1(x)\overline{Y}(x)&=&\partial_+[p^w(s^*_A(x))F_1[x]+p^p(s^*_A(x))]\nonumber
\end{eqnarray}
Similarly, we have
\begin{displaymath}
f_1(x)g(x)=\partial_-[p^w(s^*_A(x))F_1[x]+p^p(s^*_A(x))]
\end{displaymath}
Because $\overline{Y}(x)=g(x)$ for almost all $x$,
then when we do integration on both right and left derivatives, the difference vanishes, i.e.
\begin{eqnarray}
\int^{x}_{a_1}f_1(t)g(t)dt=p^w(s^*_A(x))F_1[x]+p^p(s^*_A(x))-p^p(s^*_A(a_1))\nonumber
\end{eqnarray}
\end{proof}
\textit{{\bf Remark 1.}
Up to now, we have developed a new strategy $s^*_A$ for $A$ based on $s$, with
at least 2 desirable properties:
it yields at least as much utility as $s$ and is left-hand differentiable and right-hand differentiable. In the following, we will calculate the winning probability and find out the $s^*_A$ with the optimal utility.
}

\noindent \textit{{\bf Remark 2.}
From $s$ (we only need the weakly increasing condition), we could define $g$ directly,
but $s$ cannot be calculated by $g$.
To see this,
the follower bid $s_A(t)$ and when $t=x$, the follower achieves the
largest utility.
Take first price auction for example,
we have $(g(x)-s_A(x))F_1[x]\geq(g(x)-s_A(t))F_1[t]\forall t$,
then $s_A(t)\geq g(x)-\frac{(g(x)-s_A(x))F_1[x]}{F_1[t]}$, equality can be achieved by setting $t$ to be $x$.
Moreover $s_A(t)\geq\sup_x(g(x)-\frac{u_B(x)}{F_1[t]})=\sup_xeu_B(g(x),t)$, equality may not be achieved,
because if we fix $t$ first, there maybe no corresponding $x$.
Thus we do not have the exact formula of $s_A(t)$.
If $s$ is optimal, for the leader's every type, his bid should be as small as possible without changing follower's behavior.
To do this, when $s_A(t)>\sup_xeu_B(g(x),t)$, we can decrease his bid to $\epsilon+\sup_xeu_B(g(x),t)$, without letting the follower match his bid.
This is the nature the smooth method.
So in the optimal strategy, we should have $s_A(t)=\sup_xeu_B(g(x),t)$, which is exact the smooth method, and we create $s^*_A$ to be it.
}

We can check the correctness of relationship between $g$ and $s^*_A$ in Example~\ref{eg2}.

When $x_0\in[0,0.4]$,
$$\frac1{F_1[x_0]}\int^{x_0}_{a_1} g(x)f_1(x)dx=\frac1{x_0}[\int^{x_0}_{0}x/2dx]=x_0/4=s^*_A(x_0)$$

When $x_0\in[0.4,0.65]$,
$$\frac1{F_1[x_0]}\int^{x_0}_{a_1} g(x)f_1(x)dx=\frac1{x_0}[\int^{0.4}_{0}x/2dx+\int^{x_0}_{0.4}(2x-0.3)dx]=x_0-0.3=s^*_A(x_0)$$

When $x_0\in[0.65,1]$,
\begin{eqnarray}
&&\frac1{F_1[x_0]}\int^{x_0}_{a_1} g(x)f_1(x)dx=\frac1{x_0}[\int^{0.4}_{0}x/2dx+\int^{0.65}_{0.4}(2x-0.3)dx+\int^{x_0}_{0.65}1dx]\nonumber\\
&&=\frac1{x_0}[0.04+0.65*0.35-0.04+x_0-0.65]=1-\frac{0.65^2}{x}=s^*_A(x_0)\nonumber
\end{eqnarray}

\begin{definition}
\begin{eqnarray}
O_1&=&\{ s |\textrm{~strategies resulted from any nonnegative strategy after smoothing},s_A(a_1)=0\}\nonumber\\
O_2&=&\{(g,s_A(a_1))|g \textrm{~is weakly increasing and left continuous and in} [0,b_2]\}\nonumber
\end{eqnarray}
\end{definition}
By now, definition~\ref{def7} gives a mapping $M_1: O_1 \rightarrow O_2$.
In fact, we will prove that there is a bijective mapping between the two sets.
The idea is that we construct a mapping $M_2: O_2 \rightarrow O_1$ and
proves that $M_1\circ M_2=I$.

\begin{theorem}
There is a bijective mapping between $O_1$ and $O_2$.
\label{theorem8}
\end{theorem}

As we will see in Lemma~\ref{lemmasa1=0}, the optimal strategy is in set $O_1$,
then finding optimal strategy is equivalently to finding the optimal function $g$ in set $O_2$.
\section{optimize equal-bid function g}
In this section, we solve for the optimal $g$ to derive the final form of $s^*_A$.
\subsection{General bid-based payment}
\begin{lemma}\label{lemmasa1=0}
In the optimal strategy, $s^*_A(a_1)=0$.
\end{lemma}
Now we formalize the problem:
\begin{eqnarray}
Maximize: u_A(s^*_A)&&\int_{a_1}^{a_2}\{[x-p^w(s^*_A(x))]F_2[g(x)]-p^p(s^*_A(x))\}f_1(x)dx\nonumber\\
s.t.&& g\in[0,b_2] \textrm{ is left continuous, weakly increasing.}\nonumber\\
&&\int^{x}_{a_1}f_1(t)g(t)dt-p^w(s^*_A(x))F_1[x]-p^p(s^*_A(x))=0\nonumber
\end{eqnarray}
Using the method of Lagrange multipliers, the problem becomes
\begin{eqnarray}
Maximize: u_A(s^*_A)&=&\int_{a_1}^{a_2}\{[x-p^w(s^*_A(x))]F_2[g(x)]-p^p(s^*_A(x))\}f_1(x)dx\nonumber\\
&&+\int_{a_1}^{a_2}(\int^{x}_{a_1}f_1(t)g(t)dt-p^w(s^*_A(x))F_1[x]-p^p(s^*_A(x)))R(x)dx\nonumber\\
s.t.&& g\in[0,b_2] \textrm{ is left continuous, weakly increasing.}\nonumber
\end{eqnarray}
Rewrite the objective equation to reduce the double integral of g into single integral.
We denote by $L$ the integrand.

\begin{eqnarray}
Maximize: u_A(s^*_A)&=&\int_{a_1}^{a_2}\{[x-p^w(s^*_A(x))]F_2[g(x)]-p^p(s^*_A(x))\}f_1(x)dx\nonumber\\
&&+\int_{a_1}^{a_2}\{g(x)f_1(x)\int^{a_2}_xR(t)dt-[p^w(s^*_A(x))F_1[x]+p^p(s^*_A(x))]R(x)\}dx\nonumber\\
&=&\int_{a_1}^{a_2}L(s^*_A(x),g(x))dx\nonumber\\
s.t.&& g\in[0,b_2] \textrm{ is left continuous, weakly increasing.}\nonumber
\end{eqnarray}
We say $g(x)$ is free, if $g$ increases or decreases by a small amount on point $x$, $g$ still satisfy the weak monotone constraint and boundary constraint.
In the optimal solution, when $g(x)$ is free, we should have the following equations:
(When $g(x)$ is not free, only Equation (\ref{beq1}) fails to be zero)
\begin{eqnarray}
0=\frac{\partial L}{\partial g}&=&[x-p^w(s^*_A(x))]f_2(g(x))f_1(x)+f_1(x)\int^{a_2}_xR(t)dt\label{beq1}\\
0=\frac{\partial L}{\partial s^*_A}&=&-(p^w)'(s^*_A(x))F_2[g(x)]f_1(x)-(p^p)'(s^*_A(x))f_1(x)\label{beq2}\\
&&-(p^w)'(s^*_A(x))F_1[x]R(x)-(p^p)'(s^*_A(x))R(x)\nonumber\\
0&=&\int^{x}_{a_1}f_1(t)g(t)dt-p^w(s^*_A(x))F_1[x]-p^p(s^*_A(x))\label{beq3}
\end{eqnarray}
Note that the last equation is not $\frac{\partial L}{\partial R}$,
it is the partial derivative of the original objective function with respect to $R$.

In general, we substitute Equation (\ref{beq2}) and (\ref{beq3}) into (\ref{beq1}),
and let $h(x)$ denote $\frac{\partial L}{\partial g}$. Since $g$ is left-continuous,
it's easy to prove $h$ is left-continuous.
\begin{theorem}\label{theoremshg}
(1)For an interval $L$, if $h(x)>0 x\in L$, we have
\begin{displaymath}
g(x)=\lim_{t\rightarrow (\sup L)^+}g(t)\ x\in L
\end{displaymath}
Moreover, if $\sup L=a_2$, then $g(x)=b_2, x\in L$.

Similar, if $h(x)<0, x\in L$, we have
\begin{displaymath}
g(x)=\lim_{t\rightarrow \inf L}g(t), x\in L
\end{displaymath}
If $\inf L=a_1$, then $g(x)=0, x\in L$.

(2)There is an optimal $g$ such that $g(x)\in 0 \cup [b_1,b_2]$.
\end{theorem}

In fact, $g$ can be derived explicitly in fair general settings, as we show below.

\subsection{First price auction}

In first price auction, $p^w(t)=t$ and $p^p(t)=0$.
Equation (\ref{beq1}) - (\ref{beq3}) and $h(x) $become
\begin{eqnarray}
\frac{\partial L}{\partial g}&=&[x-s^*_A(x)]f_2(g(x))f_1(x)+f_1(x)\int^{a_2}_xR(t)dt\nonumber\\
0&=&-F_2[g(x)]f_1(x)-F_1[x]R(x)\nonumber\\
s^*_A(x)&=&\frac{1}{F_1[x]}\int^{x}_{a_1}f_1(t)g(t)dt\nonumber\\
h(x)&=&f_1(x)(xf_2(g(x))-\frac{f_2(g(x))}{F_1[x]}\int^x_{a_1}f_1(t)g(t)dt-\int^{a_2}_x\frac{F_2[g(t)]f_1(t)}{F_1[t]}dt)\nonumber
\end{eqnarray}

\begin{theorem}\label{theorem7}When $F_2$ is uniformly distributed on $[b_1, b_2]$,

(1)if $\forall t, 2f^2_1(t)-F_1[t]f'_1(t)\geq0$, then optimal $g(x)$ consists of at most 3 values.

(2)
if $b_1=0$, then $g(x)$ consists of 2 pieces.
When $t_0=a_2\int^{a_2}_{t_0}\frac{f_1(t)}{F_1[t]}dt$ has a solution,
\begin{displaymath}
g(x)=\left\{
\begin{array}{ll}
0 & x\in (a_1,t_0)\\
b_2 & x\in (t_0,a_2)
\end{array}
\right.
\quad \textrm{,where } t_0=a_2\int^{a_2}_{t_0}\frac{f_1(t)}{F_1[t])}dt
\end{displaymath}
Otherwise, $g(x)=b_2,\forall x\in[a_1,a_2]$.
\end{theorem}

\begin{proof}

Let $t_0=\sup\{ t|g(t)=0\}$.
If $t_0$ does not exist, let $t_0=a_1$.

When $F_2$ is uniformly distributed, $f_2(g(x))$ does not change on $(t_0,a_2]$, we denote it $f_2$.
It's easy to see that $h(x)$ is continuous on $(t_0,a_2]$.

(1)When $t>t_0$, then by Lemma~\ref{lemmasa1=0}, $g(t)\geq b_1$, $f_2(g(t))$ is constant.
\begin{eqnarray}
h(x)&=&f_1(x)f_2\cdot(x-\frac1{F_1[x]}\int^x_{t_0}g(t)f_1(t)dt+\int^{a_2}_x\frac{-f_1(t)}{F_1[t]}(g(t)-b_1)dt)\nonumber\\
(\frac{h(x)}{f_1(x)})'\frac{F^2_1(x)}{f_1(x)f_2}&=&\frac{F^2_1(x)}{f_1(x)}+\int^x_{t_0}g(t)f_1(t)dt-b_1F_1[x]\label{equ5}\\
((\frac{h(x)}{f_1(x)})'\frac{F^2_1(x)}{f_1(x)f_2})'&=&\frac{2F_1[x]f^2_1(x)-F^2_1(x)f'_1(x)}{f^2_1(x)}+g(x)f_1(x)-b_1f_1(x)\nonumber
\end{eqnarray}
Since $\forall t, 2f^2_1(t)-F_1[t]f'_1(t)\geq0$, then $((\frac{h(x)}{f_1(x)})'\frac{F^2_1(x)}{f_1(x)f_2})'>0$,
then there is at most one cross for $(\frac{h(x)}{f_1(x)})'=0$,
then there are at most two crosses for $\frac{h(x)}{f_1(x)}=0$, i.e. $h(x)=0$.
Assume the two crosses are $t_1$ and $t_2$. There are several cases that how the sign of $h$ changes,
for any case, we can prove function $g$ consists of at most three pieces.
Take an example, we consider the following case:
\begin{displaymath}
h(x)=\left\{
\begin{array}{ll}
>0 & x\in (t_0,t_1)\\
<0 & x\in (t_1,t_2)\\
>0 & x\in (t_2,a_2)
\end{array}
\right.
\end{displaymath}
By Theorem~\ref{theoremshg} $g(x)=g(t_1), x\in(t_0,t_2]$ is constant and $g(x)=b_2, x\in(t_2,a_2]$ is constant.
So there are at most two pieces in $g(x)$ in $(t_0,a_2]$. Counting the part $g$ is zero, there are at most three pieces.

(2)
Since $b_1=0$, $f_2(0)=f_2$, which means $h$ is continuous on the whole range.

When $x\leq t_0$,
$$
\frac{h(x)}{f_1(x)}=f_2x+\int^{a_2}_{t_0}\frac{-f_1(t)}{F_1[t]}g(t)f_2dt
$$
So $(\frac{h(x)}{f_1(x)})'>0 \forall x\in(a_1,t_0]$.
Moreover, $h(x)$ is continuous at point $t_0$.
Combined Equation~(\ref{equ5}), we have $\frac{h(x)}{f_1(x)}$ strictly increases on the whole range $[a_1, a_2]$.
Then equation $h(x)=0$ has at most one solution.
When there is no solution, since $\lim_{t\rightarrow a_2}h(t)>0$, it
must be $h(x)>0$ for whole range. Then it's a special subcase of the one solution case by setting $t_0=a_1$.
So we only need to consider the nontrivial case that there is one solution for $h(x)=0$. Assume
\begin{displaymath}
h(x)=\left\{
\begin{array}{ll}
<0 & x\in (a_1,t_0)\\
>0 & x\in (t_0,a_2)
\end{array}
\right.
\end{displaymath}
By Theorem~\ref{theoremshg}, we have
$g(x)=0, x\in[a_1,t_0)$ and $g(x)=b_2, x\in(t_0,b_2]$.
Since $g(x)$ is left continuous, it should be $g(t_0)=0$.

Then we compute the optimal breaking point $t_0$, from $h(t_0)=0$, we have
\begin{eqnarray}
t_0=b_2\int^{a_2}_{t_0}\frac{f_1(t)}{F_1[t]}dt\nonumber
\end{eqnarray}
When there is no solution we set $t_0=a_1$.
In conclusion, when $t_0=a_2\int^{a_2}_{t_0}\frac{f_1(t)}{F_1[t]}dt$ has a solution,
\begin{displaymath}
g(x)=\left\{
\begin{array}{ll}
0 & x\in (a_1,t_0)\\
b_2 & x\in (t_0,a_2)
\end{array}
\right.,
\quad \textrm{where } t_0=a_2\int^{a_2}_{t_0}\frac{f_1(t)}{F_1[t]}dt
\end{displaymath}
When there is no solution, $g(x)=b_2,\forall x\in[a_1,a_2]$.
\end{proof}

\begin{theorem}
\label{theorem1st}
In Example~\ref{eg1}, the optimal utility of the leader is 0.22. The closed-form representation of the optimal $g(x)$ and $s^*_A(x)$ are:
\begin{displaymath}
g(x)=\left\{
\begin{array}{ll}
0 & x\in [0,t_0] \\
1 & x\in (t_0,1]
\end{array}\right.
\qquad
s^*_A(x)=\left\{
\begin{array}{ll}
0 & x\in [0,t_0]\\
1-\frac{t_0}{x} & \in (t_0,1] \\
\end{array} \quad ,t_0\approx0.567\right.
\end{displaymath}
Here $x_0$ is actually the solution of $t_0=b_2\int^{a_2}_{t_0}\frac{f_1(x)}{F_1[x]}dx$.
\end{theorem}

We should notice that the leaders bids zero $56.7\%$ of the time.
We also observe the following interesting, even counterintuitive property.
\begin{corollary}
The leader sometimes bids above his valuation.
\label{corollary1}
\end{corollary}

\begin{proof}
Consider the setting:
$b_1=0,b_2=10,a_1=0,a_2=2$ and $f_2(x)=0.1, x\in[0,10]$
\begin{displaymath}
f_1(x)=\left\{
\begin{array}{ll}
2/3 & x\in (0,1]\\
1/3 & x\in (1,2]
\end{array}
\right.
\end{displaymath}
We compute the optimal commitment.

By Theorem~\ref{theorem7}, check whether there is a $t_0$ such that $t_0=2\int^2_{t_0}\frac{f_1(t)}{F_1[t]}dt$.
After considering both cases $t_0\leq 1$ and $t_0>1$, we found $t_0\approx1.3386$.
Then,
\begin{displaymath}
g(x)=\left\{
\begin{array}{ll}
0 & x\in (0,t_0]\\
10 & x\in (t_0,2]
\end{array}
\right.
\end{displaymath}
Substitute in $s^*_A$, we get
$$s^*_A(2)=\int^2_{t_0}10/3dx=2.2048>2,$$
So in the optimal strategy, the leader overbids his value!
\end{proof}

The corollary is counterintuitive and deserves a close scrutiny.

Let $x_0$ denote the solution of $s^*_A(x)=x$, and $x_0\approx1.88$.
The leader's strategy $s^*_A$ and winning probability $P_A[t,s^*_A]$ is listed below,
\begin{displaymath}
(s^*_A(x),P_A[t,s^*_A])=\left\{
\begin{array}{ll}
(0,0) & x\in [0,1.338]\\
(\frac{10(x-1.338)}{x+1}\leq x, 1) & x\in(1.338,1.88]\\
(\frac{10(x-1.338)}{x+1}\geq x, 1) & x\in (1.88,2]
\end{array}
\right.
\end{displaymath}

The leader gives away positive utility when his value is low, i.e. $x \in[0,1.338]$.
Though the probability that the leader's value lies in $[0,1.338]$ is large, around $78\%$.
However, even if the leader wins, he only gains a small amount of utility, since his valuation is small.

The leader also sacrifices positive utility when $A$'s value is high, i.e. $x\in (1.88,2]$, he is supposed to have the high utility in this case.
However, the probability of the leader's value lies in $[1.88,2]$ is $4\%$.
He only loses a small amount of expected utility.

By placing the aggressive bids on $(1,338,2]$,
the leader threats the follower so that she bids zero.
We say ``aggressive'' because the leader's bid increases very fast in that interval.
When the leader has relatively high value, i.e. $x\in (1,338, 1.88]$, he wins deterministically with relative lower prices!  It turns out that the utility increment on this interval is high enough to compensate the decrements on the other two intervals.

\subsection{All-pay auction.}

In first price auction, $p^w(t)=t$ and $p^p(t)=0$.
Equation (\ref{beq1})(\ref{beq2})(\ref{beq3}) and $h(x) $become

\begin{eqnarray}
\frac{\partial L}{\partial g}&=&xf_2(g(x))f_1(x)+f_1(x)\int^{a_2}_xR(t)dt\nonumber\\
0&=&-f_1(x)-R(x)\nonumber\\
s^*_A(x)&=&\int^{x}_{a_1}f_1(t)g(t)dt\nonumber\\
h(x)&=&xf_2(g(x))f_1(x)+f_1(x)\int^{a_2}_x-f_1(t)dt\label{beq4}\\
&=&f_1(x)[xf_2(g(x))-1+F_1[x]]\nonumber
\end{eqnarray}
\begin{theorem}\label{theoremap7}
When $f_2$ is weakly increasing,
then optimal $g(x)$ is a step function consisting of at most 2 values,
0 and $b_2$, and the cut point $t_0$ of $g$ is the solution of
$b_2-t-b_2F_1[t]=0$. In particular, when $F_1$ is uniform distribution, $t_0=\frac{b_2a_2}{b_2+a_2-a_1}$.
\begin{displaymath}
g(x)=\left\{
\begin{array}{ll}
0 & x\in [a_1,t_0]\\
b_2 & x\in (t_0,a_2]
\end{array}
\right.\qquad
s^*_A(x)=\left\{
\begin{array}{ll}
0 & x\in [a_1,t_0]\\
b_2(F_1[x]-F_1[t_0]) & x\in (t_0,a_2]
\end{array}
\right.
\end{displaymath}
\end{theorem}

\begin{proof}
Let $t_0=\sup\{ t|g(t)=0\}$.
If $t_0$ does not exist, let $t_0=a_1$.

We first consider the case when $x\in([t_0,a_2]$.
When $f_2$ is weakly increasing,
\begin{displaymath}
\frac{h(x)}{f_1(x)}=xf_2(g(x))-1+F_1[x]
\end{displaymath}
strictly increases by Equation (\ref{beq4}).
Then there is at most one cross for $\frac{h(x)}{f_1(x)}=0$, i.e. $h(x)=0$.
If $h$ keeps positive, then $g(x)$ should be constant $b_2$.
If $h$ keeps negative, then $g(x)$ should be constant $0$.
Assume the cross is at $t_1$.

\begin{displaymath}
h(x)=\left\{
\begin{array}{ll}
<0 & x\in (t_0,t_1)\\
>0 & x\in (t_1,a_2)
\end{array}
\right.
\end{displaymath}

Then $g(x), x\in[t_0,t_1)$ is constant zero and $g(x), x\in[t_1,a_2]$ is constant $b_2$.
Otherwise we can decrease $g(x), x\in(t_0,t_1)$ and increase $g(x), x\in(t_1,a_2]$.
Contradicts to $g$ is optimal.

Consider the case when $x\in[a_1,t_0]$.
We have $t_0=t_1$.
So there are at most two pieces in $g(x)$ in $[a_1,a_2]$.

Now we compute optimal $t_0$.
Since the case for $g$ has only one value can be seen a sub-special case of the two values,
we only need to consider the two values case.
We first compute $s^*_A$,
\begin{displaymath}
s^*_A(x)=\int^x_{a_1}f_1(t)g(t)dt=\left\{
\begin{array}{ll}
0 & x\in [a_1,t_0)\\
(F_1[x]-F_1[t_0]) & x\in [t_0,a_2]
\end{array}
\right.
\end{displaymath}

Then the leader's expected utility is
\begin{eqnarray}
u_A(s^*_A)&=&\int^{a_2}_{a_1}f_1(x)[xF_2[g(x)]-s^*_A(x)]dx\nonumber\\
&=&\int^{a_2}_{t_0}f_1(x)[x-F_1[x]b_2+F_1[t_0]b_2]dx\nonumber\\
&=&\int^{a_2}_{t_0}f_1(x)[x-F_1[x]b_2]dx+F_1[t_0]b_2[1-F_1[t_0]]\nonumber
\end{eqnarray}

The derivative of the utility with respect to $t_0$ is
$$f_1(t_0)[b_2-t_0-b_2F_1[t_0]]$$
Then optimal $t_0$ must be the solution of $b_2-t-b_2F_1[t]=0$, and it's unique.
\end{proof}

%
%
%
%
%
%
%

\section{Acknowledgement}

We are grateful to the colleagues in Baidu for enlightening discussions on the passive bidding behaviors in sponsored search auctions. This work was supported in part by the National Basic Research Program of China Grant 2011CBA00300, 2011CBA00301, the National Natural Science Foundation of China Grant 61033001, 61361136003, 61303077, and a Tsinghua University Initiative Scientific Research Grant.


\bibliographystyle{plainnat}
\bibliography{team}

\begin{thebibliography}{25}
\providecommand{\natexlab}[1]{#1}
\providecommand{\url}[1]{\texttt{#1}}
\expandafter\ifx\csname urlstyle\endcsname\relax
  \providecommand{\doi}[1]{doi: #1}\else
  \providecommand{\doi}{doi: \begingroup \urlstyle{rm}\Url}\fi

\bibitem[Abraham et~al.(2013)Abraham, Athey, Babaioff, and
  Grubb]{abraham2013peaches}
Ittai Abraham, Susan Athey, Moshe Babaioff, and Michael Grubb.
\newblock Peaches, lemons, and cookies: designing auction markets with
  dispersed information.
\newblock In \emph{EC}, pages 7--8, 2013.

\bibitem[Aryal and Gabrielli(2013)]{aryal2013testing}
Gaurab Aryal and Maria~F Gabrielli.
\newblock Testing for collusion in asymmetric first-price auctions.
\newblock \emph{International Journal of Industrial Organization}, 31\penalty0
  (1):\penalty0 26--35, 2013.

\bibitem[B{\"o}rgers et~al.(2013)B{\"o}rgers, Cox, Pesendorfer, and
  Petricek]{borgers2013equilibrium}
Tilman B{\"o}rgers, Ingemar Cox, Martin Pesendorfer, and Vaclav Petricek.
\newblock Equilibrium bids in sponsored search auctions: Theory and evidence.
\newblock \emph{American economic Journal: microeconomics}, 5\penalty0
  (4):\penalty0 163--187, 2013.

\bibitem[Chawla and Hartline(2013)]{chawla2013auctions}
Shuchi Chawla and Jason~D Hartline.
\newblock Auctions with unique equilibria.
\newblock In \emph{Proceedings of the fourteenth ACM conference on Electronic
  commerce}, pages 181--196. ACM, 2013.

\bibitem[Conitzer and Sandholm(2006)]{commitEC06}
Vincent Conitzer and Tuomas Sandholm.
\newblock Computing the optimal strategy to commit to.
\newblock In \emph{Proceedings of the 7th ACM conference on Electronic
  commerce}, pages 621--641. ACM, 2006.

\bibitem[Edelman and Ostrovsky(2007)]{edelman2007strategic}
Benjamin Edelman and Michael Ostrovsky.
\newblock Strategic bidder behavior in sponsored search auctions.
\newblock \emph{Decision support systems}, 43\penalty0 (1):\penalty0 192--198,
  2007.

\bibitem[Edelman et~al.(2007)Edelman, Ostrovsky, and
  Sshwartz]{edelman2007internet}
Benjamin Edelman, Michael Ostrovsky, and Michael Sshwartz.
\newblock Internet advertising and the generalized second-price auction:
  Selling billions of dollars worth of keywords.
\newblock \emph{The American economic review}, 97\penalty0 (1):\penalty0
  242--259, 2007.

\bibitem[Fibich and Gavish(2011)]{fibich2011numerical}
Gadi Fibich and Nir Gavish.
\newblock Numerical simulations of asymmetric first-price auctions.
\newblock \emph{Games and Economic Behavior}, 73\penalty0 (2):\penalty0
  479--495, 2011.

\bibitem[Fibich and Gavish(2012)]{fibich2012asymmetric}
Gadi Fibich and Nir Gavish.
\newblock Asymmetric first-price auctions-a dynamical-systems approach.
\newblock \emph{Mathematics of Operations Research}, 37\penalty0 (2):\penalty0
  219--243, 2012.

\bibitem[Gray and Reiley(2013)]{Gray2013}
Sean Gray and David~H. Reiley.
\newblock {Measuring the Benefits to Sniping on eBay: Evidence from a Field
  Experiment}.
\newblock \emph{Journal of Economics and Management}, 9\penalty0 (2):\penalty0
  137--152, July 2013.

\bibitem[Hartline et~al.(2014)Hartline, Hoy, and Taggart]{hartline2014price}
Jason Hartline, Darrell Hoy, and Sam Taggart.
\newblock Price of anarchy for auction revenue.
\newblock In \emph{Proceedings of the fifteenth ACM conference on Economics and
  computation}, pages 693--710. ACM, 2014.

\bibitem[Kaplan and Zamir(2012)]{kaplan2012asymmetric}
Todd~R Kaplan and Shmuel Zamir.
\newblock Asymmetric first-price auctions with uniform distributions: analytic
  solutions to the general case.
\newblock \emph{Economic Theory}, 50\penalty0 (2):\penalty0 269--302, 2012.

\bibitem[Lebrun(1999)]{lebrun1999first}
Bernard Lebrun.
\newblock First price auctions in the asymmetric n bidder case.
\newblock \emph{International Economic Review}, 40\penalty0 (1):\penalty0
  125--142, 1999.

\bibitem[Letchford and Conitzer(2010)]{extensiveEC10}
Joshua Letchford and Vincent Conitzer.
\newblock Computing optimal strategies to commit to in extensive-form game.
\newblock In \emph{Proceedings of the 11th ACM conference on Electronic
  commerce}, pages 83--92. ACM, 2010.

\bibitem[Lopomo et~al.(2011)Lopomo, Marx, and Sun]{lopomo2011bidder}
Giuseppe Lopomo, Leslie~M Marx, and Peng Sun.
\newblock Bidder collusion at first-price auctions.
\newblock \emph{Review of Economic Design}, 15\penalty0 (3):\penalty0 177--211,
  2011.

\bibitem[Marshall and Marx(2007)]{marshall2007bidder}
Robert~C Marshall and Leslie~M Marx.
\newblock Bidder collusion.
\newblock \emph{Journal of Economic Theory}, 133\penalty0 (1):\penalty0
  374--402, 2007.

\bibitem[McAfee and McMillan(1992)]{mcafee1992bidding}
R~Preston McAfee and John McMillan.
\newblock Bidding rings.
\newblock \emph{The American Economic Review}, pages 579--599, 1992.

\bibitem[Pesendorfer(2000)]{pesendorfer2000study}
Martin Pesendorfer.
\newblock A study of collusion in first-price auctions.
\newblock \emph{The Review of Economic Studies}, 67\penalty0 (3):\penalty0
  381--411, 2000.

\bibitem[Roth and Ockenfels(2002)]{Roth2002}
Alvin~E. Roth and Axel Ockenfels.
\newblock {Last-Minute Bidding and the Rules for Ending Second-Price Auctions:
  Evidence from eBay and Amazon Auctions on the Internet}.
\newblock \emph{American Economic Review}, 92\penalty0 (4):\penalty0
  1093--1103, September 2002.

\bibitem[Skreta(2006)]{Skreta}
Vasiliki Skreta.
\newblock Sequentially optimal mechanisms.
\newblock \emph{Review of Economic Studies}, pages 1085--1111, 2006.

\bibitem[Tambe(2011)]{Tambe2011}
Milind Tambe.
\newblock \emph{Security and Game Theory: Algorithms, Deployed Systems, Lessons
  Learned}.
\newblock Cambridge University Press, New York, NY, USA, 1st edition, 2011.
\newblock ISBN 1107096421, 9781107096424.

\bibitem[Vickrey(1961)]{Vic61}
W.~Vickrey.
\newblock {Counterspeculation, Auctions and Competitive Sealed Tenders}.
\newblock \emph{Journal of Finance}, pages 8--37, 1961.

\bibitem[von Stackelberg(1934)]{Stackelberg}
Heinrich von Stackelberg.
\newblock \emph{Marktform und Gleichgewicht}.
\newblock Springer, 1934.

\bibitem[von Stengel and Zamir(2004)]{leader}
Bernhard von Stengel and Shmuel Zamir.
\newblock Leadership with commitment to mixed strategies.
\newblock \emph{CDAM Research Report LSE-CDAM-2004-01}, 2004.

\bibitem[Xu and Ligett(2014)]{xuyunjian2014}
Yunjian Xu and Katrina Ligett.
\newblock Commitment in first-price auctions.
\newblock Technical report, 2014.

\end{thebibliography}

\newpage
\begin{appendix}
\section{Appendix: Omitted Proofs}

Proof of \textbf{Lemma~\ref{lemma0}}.

Suppose otherwise $u_B(y)$ cannot be attained. Since the domain is bounded, by the definition of $u_B(y)$,
there exists $t_0$, such that:
\begin{displaymath}
\exists \{t_n\}\rightarrow t_0,\ s.t.\ \lim [y\cdot P_B[t_n]-p^p(t_n)-p^w(t_n)\cdot P_B[t_n]]=u_B(y)
\end{displaymath}
By tie-breaking rule, we have $\lim P_B[t_n] \leq P_B[t_0]$. Since $u_B(y)\geq 0$, we have
$\lim (y-p^w(t_n))\geq 0$.

$$u_B(y)=\lim [(y-p^w(t_n)) P_B[t_n]-p^p(t_n)] \leq (y-p^w(t_0))\cdot P_B[t_0]-p^p(t_0)\leq u_B(y)$$
Hence, $u_B(y)=(y-p^w(t_0))\cdot P_B[t_0]-p^p(t_0)$, i.e. $u_B(y)$ can be attained by bid $t_0$.

Next we prove the smallest best response exists. Suppose otherwise there is no smallest best response, then among all the best responses, there exists $\underline{t}$, such that
$$\exists \{t_n\}\rightarrow \underline{t}\ s.t. \ u_B(y)=\lim y\cdot P_B[t_n]-p^p(t_n)-p^w(t_n)\cdot P_B[t_n] \quad\forall n$$
By same argument as above, we know
$$u_B(y)=y\cdot P_B[\underline{t}]-p^p(\underline{t})-p^w(\underline{t})\cdot P_B[\underline{t}]$$ i.e. $\underline{t}$ is the smallest best response, contradiction.
$\Box$

Proof of \textbf{Lemma~\ref{lemma1}}.

It is important to note that, from the follower's perspective, she only cares about the overall distribution of the leader's bids induced by his strategy: as long as the distribution of A's bids is unchanged, B's best response remain unchanged. Our idea is then to rearrange (sort) the leader's bids without changing the underlying distribution.

For any strategy $s_A$, we fix the bids distribution $D$, i.e. the distribution of $s_A(v),\ v\sim F_1$,
and rematch the leader's valuations to bids and create some new strategy $\breve{s}_A$.
So the distribution of $\breve{s}_A(v),\ v\sim F_1$, is same as distribution of $s_A(v),\ v\sim F_1$.
In this process, the follower's best response remain unchanged, hence (2).
In addition, her bid distribution is also fixed.
As a result, for any single bid of $A$, the probability of winning is also the same. As a result, the overall winning probability is also unchanged.

Now look at the rank-and-bid based payment function, the leader's expected payment is:
\begin{eqnarray}
&&\int_{t\sim D}Pr[t,s_A]p^w(t)+p^p(t)dt\nonumber\\
&=&\int_{t\sim D}Pr[t,s_A]p^w(t)+p^p(t)dt\nonumber
\end{eqnarray}
also unchanged after rematching.

To improve the leader's expected utility, which equals expected social welfare minus expected payment (fixed)
we only need to increase the expected social welfare, given that the overall winning probability is fixed.
Therefore, in rematching, we sort the strategy monotonically such that higher valuation with higher winning probability, i.e., higher bid. In this way, we guarantee the total amount of fixed winning probability is allocated to the highest types, thus yields the highest expected social welfare. From the reasoning above, we conclude that the leader's expected utility weakly increases, hence (3).

The above process can be thought of as rematching a bid that is in the top $q$ quantile of the bid distribution to a type that is in the top $q$ of the type distribution, for all $q$.
So there is no mixed strategy at any type, i.e., $\breve{s}_A(v)$ is deterministic $\forall v$, hence (1).
%
$\Box$

Proof of \textbf{Theorem~\ref{lemma2}}.

According to the definition of $u$:
\begin{eqnarray}
&&u_B(y_1)=y_1\cdot P_B[a]-p^p(a)-p^w(a)\cdot P_B[a]\geq y_1\cdot P_B[b]-p^p(b)-p^w(b)\cdot P_B[b]\label{eqa3}\\
&&u_B(y_2)=y_2\cdot P_B[b]-p^p(b)-p^w(b)\cdot P_B[b]\geq y_2\cdot P_B[a]-p^p(a)-p^w(a)\cdot P_B[a]\label{eqa4}\\
&\Rightarrow&(y_2-y_1)[P_B[b]-P_B[a]]\geq 0\textrm{\qquad (\ref{eqa3})+(\ref{eqa4})}\nonumber\\
&\Rightarrow& P_B[b]=P_B[a] \textrm{\qquad (\ref{eqa3}) and (\ref{eqa4}) become equalities.}\nonumber
\end{eqnarray}
Substitute the last equality into the Equality (\ref{eqa3}), we get $p^p(a)+p^w(a)P_B[a]=p^p(b)+p^w(b)P_B[b]$.
Because $p^p(a)+p^w(a)>p^p(b)+p^w(b)$, it must be $P_B[b]=P_B[a]=0$.
\begin{eqnarray}
a,b<s_A(x)\ \forall x>a_1 \label{eqa8}
\end{eqnarray}
Furthermore $u_B(y_1)=u_B(y_2)=0$ and $p^p(a)=p^p(b)$.

Since we can set $a$ to be any element in $S_B(y_1)$ as long as $a>b$, then we have
$\forall a\ s.t.\ b<a\in S_B(y_1)$, we have $P_B[a]=0$ and $p^p[a]=0$. Since $P_B$ and $p^p$ are weakly increasing, then they are true for all $a\in S_B(y_1)$.
So the follower can achieve the largest utility 0 by bidding $a$, i.e. $a\in S_B(y_2)$.
Hence $S_B(y_1)\subseteq S_B(y_2)$.
$\Box$

Proof of \textbf{Lemma~\ref{lemma3}}.

We first prove $u_B(y)$ is weakly increasing. For any $y_1<y_2$, if $u_B(y_1)=u_B(y_2)=0$, the lemma is correct. Otherwise,
pick $b_1\in S_B(y_1)$, $b_2\in S_B(y_2)$.
According to Theorem~\ref{lemma2}, we have $b_1\leq b_2$. By definition of $u$, we have
$$u_B(y_2)\geq(y_2-p^w(b_1))P_B[b_1]-p^p(b_1)\geq (y_1-p^w(b_1))P_B[b_1]-p^p(b_1)=u_B(y_1)$$
So, $u_B(y)$  is weakly increasing.
Next we prove the continuity. For any $y_1<y_2$, we have
\begin{eqnarray}
&&u_B(y_2)-u_B(y_1)\nonumber\\
&\leq& (y_2-p^w(b_2))P_B[b_2]-p^p(b_2)-(y_1-p^w(b_2))P_B[b_2]+p^p(b_2)\nonumber\\
&=&(y_2-y_1)P_B[b_2]\nonumber\\
&\leq&y_2-y_1\nonumber
\end{eqnarray}
For $\forall \epsilon>0$, as long as $y_2\in(y_1-\epsilon, y_1+\epsilon)$, we have $|u_B(y_2)-u_B(y_1)|<\epsilon$.
Thus, $u_B(y)$ is a continuous function.
$\Box$

Proof of \textbf{Lemma~\ref{lemma3.25}}.

First, notice that $p^w(t)a+p^p(t)$ strictly increases, so solution $t$ exists.

Second, $t(a,b)$ is continuous. Suppose not, let $t(a,b)$ jumps at point $(a,b)$.
Then $ap^w(t)+p^p(t)+b$ jumps at $(a,b)$, it could not always be zero, contradiction.

Third, $t(a,b)$ is differentiable. Suppose otherwise, $t(a,b)$ do not have partial derivatives with respect to $a$, at point $(a_0,b_0)$.
Then
\begin{eqnarray}
&&\exists k_1>k_2, \{\tilde{a}_i\},\{a_i\}\rightarrow a_0\nonumber\\
s.t. && t(\tilde{a}_i,b_0)\geq t(a_0,b_0)+k_1(\tilde{a}_i-a_0)\nonumber\\
&& t(a_i,b_0)\leq t(a_0,b_0)+k_2(a_i-a_0)\nonumber
\end{eqnarray}
\begin{eqnarray}
0&\geq&p^w(t(a,b_0)+k_1(\tilde{a}_i-a))\tilde{a}_i+p^p(t(a)+k_1(\tilde{a}_i-a_0))+b_0\label{neweq1}\\
0&=&p^w(t(a_0,b_0))a_0+p^p(t(a_0,b_0))+b_0\label{neweq2}\\
0&\geq&p^w(t(x)+k_1(x_i-x))F_1[x_i]+p^p(t(x)+k_1(x_i-x))\qquad (\ref{neweq1})-(\ref{neweq2})\nonumber\\
0&\geq& p^w(t(a,b_0)+k_1(\tilde{a}_i-a))\tilde{a}_i-p^w(t(a_0,b_0))a_0\nonumber\\
&&p^p(t(a)+k_1(\tilde{a}_i-a_0))-p^p(t(a_0,b_0))\nonumber
\end{eqnarray}
Divide $a_i-a_0$ on both sides and consider the limitation when $i$ approaches infinity, we have
\begin{eqnarray}
0\geq (p^w)'(t(a_0,b_0))k_1a_0+p^w(t(a_0,b_0))+(p^p)'(t(a_0))k_1\nonumber
\end{eqnarray}
Similarly, we have
\begin{eqnarray}
0\leq (p^w)'(t(a_0,b_0))k_2a_0+p^w(t(a_0,b_0))+(p^p)'(t(a_0))k_2\nonumber
\end{eqnarray}
These two equations together contradicts to the fact $k_1>k_2$ and $(p^w)'+(p^p)'>0$. So $t(a,b)$ is differentiable.

We differentiate $b+ap^w(t)+p^p(t)=0$ on $a$ and $b$, we have
\begin{eqnarray}
(p^p)'(t)t'_a+(p^w)'(t)t'_aa+p^w(t)=0\nonumber\\
1+(p^w)'(t)t'_ba+p^w(t)t_b'=0\nonumber
\end{eqnarray}
Since $t$,$(p^w)'$ and $(p^p)'$ is continuous, it's easy to see these two derivatives are continuous.
$\Box$

Proof of \textbf{Lemma~\ref{lemma3.5}}.

(1)Let $a=F_1[x]$, $b=u_B(y)-F_1[x]y$, $t=eu_B(a,b)$, then $a,b$ is differentiable in $x$, and $b$ is continuous in $y$.
Fix $y$ and let $t(x)=eu_B(y,x)$, then by Lemma~\ref{lemma3.25},
\begin{eqnarray}
\frac{\partial eu}{\partial x}(y,x)&=&\frac{\partial t}{\partial a}(a,b)\cdot \frac{\partial a}{\partial x}(x,y)+\frac{\partial t}{\partial b}(a,b)\cdot \frac{\partial b}{\partial x}(x,y) \nonumber\\
&=&=\frac{1}{a(p^w)'(eu)+(p^p)'(eu)}[-p^w(eu)f_1(x)+(-1)(-f_1(x))y]\nonumber\\
&=&\frac{f_1(x)y-p^w(eu_B(y,x))f_1(x)}{(p^w)'(eu_B(y,x))F_1[x]+(p^p)'(eu_B(y,x))}\nonumber
\end{eqnarray}

$\frac{\partial eu}{\partial x}(y,x)$ is continuous for both two arguments $y$ and $x$ when $f_1,(p^w)',(p^p)'$ are continuous.

If $eu_B(y,x)$ has a breaking point $y_0$, then consider the definition of $eu_B(y,x)$
$u_B(y)-F_1[x]y=-p^w(eu_B(y,x))F_1[x]-p^p(eu_B(y,x))$.
The right hand side has a breaking point $y_0$, while the left hand side is continuous, contradiction.
So $eu_B(y,x)$ is continuous in $y$.
(2)If $\exists x,y\ s.t.\ eu_B(x,y)>s_A(x)$, then consider the follower with type $y$ bids $s_A(x)$:
\begin{eqnarray}
u_B(y)&\geq& (y-p^w(s_A(x)))F_1[x]-p^p(s_A(x))\nonumber\\
&>&yF_1[x]-p^p(eu_B(y,x))-p^w(eu_B(y,x))F_1[x]\nonumber
\end{eqnarray}
which contradicts to the definition of $eu_B(x,y)$.
$\Box$

Proof of \textbf{Lemma~\ref{lemma4}}.

(1)First we prove $(x,s^*_A(x))$ must lie on some EU line. Suppose not, there exists a series number $\{y_n\}\rightarrow y_0,\ s.t.\ s^*_A(x)=\lim_{n\rightarrow \infty} eu_B(y_n,x)$. Then
$$u_B(y_n)=F_1[x](y_n-p^w(eu_B(y_n,x)))-p^p(eu_B(y_n,x))$$
We choose the limitation and get
$$u_B(y_0)=F_1[x](y_0-p^w(s^*_A(x)))-p^p(s^*_A(x))$$
So $(x,s^*_A(x))$ is on EU line $eu_B(y_0,\cdot)$.

(2) Because $eu_B(y,\cdot)$ weakly increases and $s^*_A(x)=\sup_{y\in[0,b_2]}eu_B(y,x)$. So $s^*_A(x)$ weakly increases.

For any $x_1<x_2$, s.t. $s^*_A(x_1)=s^*_A(x_2)$. Suppose $(x_1,s^*_A(x_1))$ lies on
$eu_B(y_1,\cdot)$.
Since $eu_B(y,\cdot)$ weakly increases, $eu_B(y_1,x_1)=eu_B(y_1,x_2)=s^*_A(x_1)$.
Substitute in Equation~(\ref{eqa5}), we get  $p^w(s^*_A(x_1))=y_1$, $p^p(s^*_A(x_1))=0$ and $u_B(y_1)=0$.

Furthermore, we get $eu_B(y_1,\cdot)=s^*_A(x_1)=s^*_A(x_2)$.
Because $s^*_A$ weakly increases, then $s^*_A(x)=s^*_A(x_1) \forall x\in(a_1,x_2]$,
i.e., $s^*_A(x_1)=lim_{t\rightarrow a_1}s^*_A(t)$.

Hence if $s^*_A(x)> \lim_{t\rightarrow a_1}s^*_A(t)$, $s^*_A(x)$ strictly increases.

In particular, $u_B(s^*_A(x))=0$.
Moreover $s^*_A(x_1)$ is always the solution of Equation~(\ref{eqa5}) for $y_1$ and $x\in(a_1,x_1]$.
Because $s^*_A(x)=\sup_y ue(y,x)\geq eu_B(y_1,x)=s^*_A(x_2), x\in(a_1,x_2]$
so $s^*_A(x_2)=\lim_{t\rightarrow a_1}s^*_A(t)$. That is to say when $s^*_A(x)>\lim_{t\rightarrow a_1}s^*_A(t)$, $s^*_A(x)$ strictly increases.

(3)When $s^*_A(x)=\lim_{t\rightarrow a_1}s^*_A(t)$,
$\exists x_1<x_2=x$ such that $s^*_A(x_1)=s^*_A(x_2)$.
Then from (2.1), we know $s^*_A(x)$ lies on $eu_B(y_1,\cdot)$, where $y_1=p^w(s^*_A(x_1))$.
Moreover, $u_B(y_1)=0$ and $p^p(y_1)=0$.

(4)
By Lemma~\ref{lemma3.5}, we know $eu_B(y,x)\leq s_A(x)$. We get $s^*_A(x)=\sup_y eu_B(y,x)\leq s_A(x)$ directly.

(5)At last we prove $s^*_A(x)$ is continuous. Suppose not, there is a breaking point $x$ and a difference $d>0$, s.t. $\forall \epsilon>0$, $s^*_A(x+\epsilon)-s^*_A(x_2-\epsilon)>d$. Say $(x+\epsilon,s^*_A(x+\epsilon))$ lies on curve $eu_B(y(\epsilon),x)$. Then
\begin{eqnarray}
u_B(y(\epsilon))&=&y(\epsilon)F_1[x+\epsilon]-p^p(s^*_A(x+\epsilon))-p^w(s^*_A(x+\epsilon))F_1[x+\epsilon]\label{eqa6}\\
u_B(y(\epsilon))&=&y(\epsilon)F_1[x-\epsilon]-p^p(eu_B(y(\epsilon),x-\epsilon))-p^w(eu_B(y(\epsilon),x-\epsilon))F_1[x-\epsilon]\nonumber\\
&>&y(\epsilon)F_1[x-\epsilon]-p^p(s^*_A(x+\epsilon)-d)-p^w(s^*_A(x+\epsilon)-d)F_1[x-\epsilon]\label{eqa7}
\end{eqnarray}
(\ref{eqa6})-(\ref{eqa7}), we get
\begin{eqnarray}
y(\epsilon)(F_1[x+\epsilon]-F_1[x-\epsilon])&>&p^p(s^*_A(x+\epsilon))+p^p(s^*_A(x+\epsilon)-d)\nonumber\\
&&+[p^w(s^*_A(x+\epsilon))+p^w(s^*_A(x+\epsilon)-d)]F_1[x-\epsilon]\label{equ8}
\end{eqnarray}
When $\epsilon\rightarrow 0$, lhs of (\ref{eqa7}) approaches zero, but rhs is strictly larger than zero.
So $s^*_A(x)$ is continuous.
$\Box$

Proof of \textbf{Lemma~\ref{lemma6}}.

(1)On one side, because $s^*_A(x)\leq s_A(x)$, the follower's utility does not decrease no matter what the follower's value is.
On the other side, because $s^*_A(x)\geq eu_B(y,x)$, the follower's utility does not increase when the follower's value is $y$. Since the inequality holds for any $y$, the follower's utility does not increase no matter what the follower's value is. So the follower's utility is the same when the leader's strategy is $s_A(x)$ or $s^*_A(x)$.

(2)Pick $\forall t\in S_B(y)$, we have
\begin{displaymath}
u_B(y)=(y-p^w(t))P_B[t]-p^p(t)\leq(y-p^w(t))P^*_B[t]-p^p(t)\leq u_B(y)
\end{displaymath}
The first inequality is because $s\geq s^*_A$. The second inequality is because the follower's utility is still $u_B(y)$ when the leader adopts $s^*_A$.
Then these two inequalities are actually equalities. Hence $t\in S_B^*(y)$ and $S_B(y)\subseteq S^*_B(y)$.

Moreover, we have $(y-p^w(t))P_B[t]=(y-p^w(t))P_B^*[t]$. If $P_B[t]\neq P_B^*[t]$, it must be $y=p^w(t)$. Since $u_B(y)\geq 0$, we have $p^p(t)=0$ and thus $u_B(y)=0$.

If $t>\lim_{x\rightarrow a_1}s^*_A(x)$, we can bid $\tilde{t}$ which is a little smaller than $t$,
s.t. $P_B^*[\tilde{t}]>0$, $p^w[\tilde{t}]<y$,$p^p(\tilde{t})=0$, then
the follower can achieve positive utility by bidding $\tilde{t}$, i.e.$u_B(y)>0$, contradiction.
Hence $t=\lim_{x\rightarrow a_1}s^*_A(x)$.
$\Box$

Proof of \textbf{Lemma~\ref{lemma5}}.

(1)$P^*_B[s^*_A(x_0)]\geq F_1[x_0]$, the follower with value $y_0$ can achieve $u_B(y_0)$ by bidding $s^*_A(x_0)$. So $s^*_A(x_0)\in S^*_B(y_0)$.

(2) \begin{eqnarray}
u_B(y_0)&=&(y_0-p^w(s^*_A(x_0)))P^*_B[s^*_A(x_0)]-p^p(s^*_A(x_0))\nonumber\\
&=&(y_0-p^w(s^*_A(x_0)))F_1[x_0]-p^p(s^*_A(x_0))\nonumber
\end{eqnarray}
So $eu_B(y_0,x_0)=s^*_A(x_0)$ by definition.
$\Box$

Proof of \textbf{Lemma~\ref{sconstant}}.

Suppose otherwise,
$\exists x_1<\hat{x}$, such that $s_A(x_1)\neq s^*_A(\hat{x})$.
By Lemma~\ref{lemma4}, we know it must be $s_A(x_1)>s^*_A(\hat{x})$.

By Lemma~\ref{lemma6}, we have
$eu_B(y_0,\cdot)=s^*_A(\hat{x})$,
$u_B(y_0)=0$, where $y_0=p^w(s^*_A(\hat{x}))$. For any $\epsilon>0$, we have
$u_B(y_0+\epsilon)>0$.

Let $\eta=s_A(x_1)-s^*_A(\hat{x})$.
By Lemma~\ref{lemma3.5}, $eu_B(y,a_2)$ is continuous in $y$.
So as long as $\epsilon$ is small enough, we have
$eu_B(y_0+\epsilon,a_2)<eu_B(y_0,a_2)+\eta=s_A(x_1)$

Pick any $t\in S_B[y_0+\epsilon]$, we have
\begin{eqnarray}
u_B(y_0+\epsilon)&=&(y_0+\epsilon-p^w(t))P_B[t]-p^p(t)\nonumber\\
&\leq&y_0+\epsilon-p^w(t)-p^p(t)\nonumber
\end{eqnarray}
On the other side, by definition of $eu_B(y_0+\epsilon,\cdot)$, we have
$y_0+\epsilon-p^w(eu_B(y_0+\epsilon,a_2))-p^p(eu_B(y_0+\epsilon,a_2))=u_B(y_0+\epsilon)$
Then we get $t\leq eu_B(y_0+\epsilon,a_2)<s_A(x_1)$.

By Theorem~\ref{lemma2}, since $s^*_A(\hat{x})\in S^*_B(y_0)$, $t\in S_B(y_0+\epsilon)\subseteq S^*_B(y_0+\epsilon)$,
we have $s^*_A(\hat{x})\leq t$.

Hence we have $P_B[t]\leq F_1[\hat{x}]$, then reconsider the utility when bidding $t$:
\begin{eqnarray}
u_B(y_0+\epsilon)&=&(y_0+\epsilon-p^w(t))P_B[t]-p^p(t)\nonumber\\
&\leq&(y_0+\epsilon-p^w(t))F_1[x_1]-p^p(t)\nonumber\\
&<&(y_0+\epsilon-p^w(s^*_A(\hat{x})))F_1[\hat{x}]-p^p(s^*_A(\hat{x}))\nonumber\\
&\leq&u_B(y_0+\epsilon)\nonumber
\end{eqnarray}
Contradiction. So the supposition is wrong.
$\Box$

Proof of \textbf{Lemma~\ref{lemma7}}.

For $A$'s any value $x_0>0$, we consider the utility change between strategy $s$ and strategy $s^*_A$. We prove the number of types, at which A's utility decreases, is countable, so the loss on these values is negligible and the total expected utility does not decrease.

(1)When $s^*_A(x_0)=lim_{t\rightarrow a_1}s^*_A(t)$ and $x_0\neq \hat{x}$.

By Lemma~\ref{lemma4}, $(x_0,s^*_A(x_0))$ lies on $eu_B(y_0,\cdot)$,
where $y_0=p^w(s^*_A(x_0))$ and $u_B(y_0)=0$.

By Lemma~\ref{sconstant}, we have
\begin{displaymath}
\left\{ \begin{array}{ll}
s_A(x)=s^*_A(x)=s^*_A(\hat{x}) & x<\hat{x}\\
s_A(x)\geq s^*_A(x)>0 & x>\hat{x}
\end{array} \right.
\end{displaymath}

(1.1)When the follower's type $y\leq y_0$.

We have $u_B(y)=0$, $p^w(s^*_A(\hat{x}))>y$, and $P_B(s^*_A(\hat{x}))=P_B^*(s_A(\hat{x}))>0$.
If the follower bids $s^*_A(\hat{x})$, he will get negative utility.
So $s^*_A(\hat{x})\notin S_B(y)$. Since $P_B[s^*_A(\hat{x})]=P^*_B[s^*_A(\hat{x})]=F_1[x]>0$
by Theorem~\ref{lemma2}, $s^*_A(\hat{x})>S^*_B(y)$, and $s^*_A(\hat{x})>S_B(y)$.

So the leader wins the good with value $x_0$ when the follower has value $y< y_0$ in both $s$ and $s^*_A$ strategies.

(1.2)When the follower's type $y>y_0$.

When using $s^*_A$, $\forall y>y_0$, the follower can just bid $s^*_A(x_0)$ then achieve
a positive utility, so $u_B(y)>0$.

By Theorem~\ref{lemma2}, we have
$S^*_B(y)\geq s^*_A(\hat{x})$
Then $S^*_B(y)\geq s^*_A(\hat{x})$.
So the follower with value $y>y_0$ wins the good when the leader has value $x_0$ in both $s$ and $s^*_A$ strategies.

Combined (1.1) and (1.2), the leader's winning probability is always $F_2[y_0]$.

(2) When $s^*_A(x_0)\neq \lim_{t\rightarrow a_1}s^*_A(t)$

Let $eu_B(y_0,\cdot)$ be an eu curve that contain point
$(x_0,s^*_A(x_0))$.
Then $s^*_A(x_0)\in S^*_B(y_0)$, by Lemma~\ref{lemma5}
Let $eu_B(y_1,\cdot)$ be a eu line that does not contain point
$(x_0,s^*_A(x_0))$.

(2.1) When $y_1<y_0$,

(2.1.1)We want to prove when $A$ adopts $s^*_A$, $B=y_1$ always loses against $A=x_0$.
Otherwise, $\exists t\in S^*_B(y_1)\ s.t.\ s^*_A(x_0)\leq t$.

If $s^*_A(x_0)<t$.
$P^*_B[s^*_A(x_0)]=P^*_B[t]$ by Equation~(\ref{eqa8}) in Theorem~\ref{lemma2}.
Because $s^*_A(x_0)$ strictly increases and $s^*_A(x_0)<t$, contradiction.

If $s^*_A(x_0)=t$.
Then $s^*_A(x_0)\in S^*_B(y_1)$.
By Lemma~\ref{lemma5}, we have
$(x_0,s^*_A(x_0))$ lies on $eu_B(y_1,\cdot)$, contradiction.

(2.1.2)We want to prove when $A$ adopts $s$, $B=y_1$ always loses against $A=x_0$.
Otherwise, $\exists t\in S_B(y_1)\ s.t.\ t\geq s_A(x_0)$.

By Lemma~\ref{lemma5}, $s^*_A(x_0)\notin S^*_B(y_1)$(o.w. $(x_0,s^*_A(x_0))$ lies on $eu_B(y_1,\cdot)$).

Then $s^*_A(x_0)\notin S_B(y_1)$ by Lemma~\ref{lemma6}. Thus $t\neq s^*_A(x_0)$.
Since $t\geq s_A(x_0)\geq s^*_A(x_0)$, we have $t>s^*_A(x_0)$.

Notice that $t\in S^*_B(y_1)$ and $s^*_A(x_0)\in S^*_B(y_0)$, we have
$P^*_B[t]=P^*_B[s^*_A(x_0)]$ by Theorem~\ref{lemma2}.
That contradicts to $t>s^*_A(x_0)>\lim_{t\rightarrow a_1}s^*_A(t)$.

Combined (2.1.1) and (2.1.2), the leader with $x_0$ beats follower with $y\leq y_0$ in both $s^*_A$ and $s$.

(2.2) When $y_1>y_0$.

(2.2.1)
We want to prove when $A$ adopts $s^*_A$, $B=y_1$ always beats $A=x_0$.
By Lemma~\ref{lemma5}, we have
$s^*_A(x_0)\in S^*_B(y_0)$ and $s^*_A(x_0)\notin S_B^*(y_1)$.
By Lemma~\ref{lemma6}, we have
$s^*_A(x_0)\notin S_B(y_1)$.

Since $y_1>y_0$ and $S^*_B(y_0)\nsubseteq S^*_B(y_1)$,
we have $S^*_B(y_1)\geq s^*_A(x_0)$ \footnote{Here we mean any element in set $S^*_B(y_1)$ is larger than $s^*_A(x_0)$.}by Theorem~\ref{lemma2}.
So the follower with value $y_1$ beats the leader with value $x_0$ in strategy $s^*_A$.

(2.2.2)
We want to prove when $A$ adopts $s$, $B=y_1$ always beats $A=x_0$.
Otherwise, $\exists t\in S_B(y_1)\ s.t.\ t<s_A(x_0)$.
Continue the proof above, since $S_B(y_1)\subseteq S^*_B(y_1)$, we have
$S_B(y_1)>s^*_A(x_0)$.  We then get
$$F_1[x_0]<P^*_B[t]=P_B[t]\leq F_1[x_0]$$
A contradition. The first inequality is for $t>s^*_A(x_0)$. The second equality is for Lemma~\ref{lemma6}.
The third inequality is for $t<s_A(x_0)$.

Combined (2.2.1) and (2.2.2), when $y_1>y_0$,
the leader with $x_0$ loses against the follower with $y_1$ in both $s^*_A$ and $s$.

(3)When $s^*_A(x_0)\neq \lim_{t\rightarrow a_1}s^*_A(t)$ and lies on a unique $eu_B(y_0,\cdot)$, then the leader's winning probability is same in both $s^*_A$ and $s$ strategies.

When $s^*_A(x_0)= \lim_{t\rightarrow a_1}s^*_A(t)$ and
$x_0\neq \sup\{x|s^*_A(x)= \lim_{t\rightarrow a_1}s^*_A(t)\}$
then the leader's winning probability is same in both $s^*_A$ and $s$ strategies.

In these two cases, because of $s^*_A\leq s$ and the same winning probability,
the expected utility of the leader weakly increases when changing from $s$ to $s^*_A$.

In other cases, the leader's expected utility might decrease, but however the loss on all these points is negligible.

Define $V_1$ as follows, we only need to prove $\#V_1$, the size of set, is countable.
\begin{eqnarray}
V_1=\{(x,s^*_A(x))|\textrm{point p} (x,s^*_A(x)) \textrm{lies on at least two eu lines, and\ } s^*_A(x)> \lim_{t\rightarrow a_1}s^*_A(t)\}\nonumber
\end{eqnarray}

Pick arbitrary $x_0\in V_1$, let $eu_B(y_1,\cdot)$ and $eu_B(y_2,\cdot)$ where $y_1<y_2$, be
two eu lines that pass point $(x_0,s^*_A(x_0))$. By Lemma~\ref{lemma5},
$s^*_A(x_0)\in S^*_B(y_1), S^*_B(y_2)$.

For any $y_3>y_1$,
\begin{eqnarray}
u_B(y_3)&\geq&(y_3-p^w(s^*_A(x_0)))F_1[x_0]-p^p(s^*_A(x_0))\nonumber\\
&>&(y_1-p^w(s^*_A(x_0)))F_1[x_0]-p^p(s^*_A(x_0))=u_B(y_1)\nonumber
\end{eqnarray}
So $u_B(y_3)>0$. Then consider $\forall y_3\in Q s.t. y_1<y_3<y_2$, we have
$u_B(y_2),u_B(y_3)>0$.

By Theorem~{\ref{lemma2}}, we have
$$S_B(y_1)\leq S_B(y_3)\leq S_B(y_2)$$
So $S_B(y_3)=\{s^*_A(x_0)\}$, i.e., $s^*_A(x_0)$ is the unique best response of $B=y_3$.
Now we can map any element in $V_1$ to a rational number. Since
$s^*_A(x_0)$ is the unique best response of $B=y_3$, this mapping is injective. Thus
$\#V_1$ is countable. Because the leader has continuous and no mass point distribution,
the loss on countable point is negligible.
$\Box$

Proof of \textbf{Lemma~\ref{lemma8}}.

(1)If $Y(x)$ is not closed, then $\exists \{y_n\}\rightarrow y$ such that
\begin{displaymath}
u_B(y_n)=(y_n-p^w(s^*_A(x)))F_1[x]-p^p(s^*_A(x))
\end{displaymath}
When $n$ approaches infinity, we get
\begin{displaymath}
u_B(y)=(y-p^w(s^*_A(x)))F_1[x]-p^p(s^*_A(x))
\end{displaymath}
Then $eu_B(y,\cdot)$ passes point $(x,s^(x))$, so $y\in Y(x)$.

(2)If $\exists y_1<\hat{y}$ such that $y_1\in Y(x)$, then
\begin{eqnarray}
0&=&u_B(y_1)=(y_1-p^w(s^*_A(x)))F_1[x]-p^p(s^*_A(x))\nonumber\\
&<&(\hat{y}-p^w(s^*_A(x)))F_1[x]-p^p(s^*_A(x))\leq u_B(\hat{y})\nonumber
\end{eqnarray}
Which contradicts to $u_B(\hat{y})=0$.

(3)Suppose not, then $\exists y_1\in Y(x_1), y_2\in Y(x_2)$, such that $y_1>y_2$.
Now, we have
\begin{eqnarray}
s^*_A(x_1)&\in& S_B^*(y_1)\nonumber\\
s^*_A(x_2)&\in& S_B^*(y_2)\nonumber\\
s^*_A(x_1)&<& s^*_A(x_2)\nonumber
\end{eqnarray}
By equation (\ref{eqa8}) in Theorem~\ref{lemma2}, we have
$P_B^*[s^*_A(x_1)]=P_B^*[s^*_A(x_2)]=0$, which contradicts to
$P_B^*[s^*_A(x_1)]\geq F_1[x_1]$.

(4)Pick any $y>\hat{y}$, let $t\in S_B^*(y)$.
Obviously, $t\leq s^*_A(a_2)$.
Since $u_B(y)>0$, we have $P^*_B(t)>0$, which leads $t\geq \lim_x s^*_A(x)$.

(4.1)When $t>\lim_x s^*_A(x)$.

Because $s^*_A$ is continuous and $t\leq s^*_A(a_2)$,
there exists $x_0$ such that $s^*_A(x_0)=t\in S^*_B(y)$.
By Lemma~\ref{lemma5}, we have $(x_0,s^*_A(x_0))$ lies on $eu_B(y,\cdot)$.

(4.2)When $t=\lim_x s^*_A(x)$.

Let $x_0=F_1^{-1}[P_B^*[t]]$, then $s^*_A(x_0)=t$.
\begin{eqnarray}
u_B(y)&=&(y-p^w(t))P^*_B[t]-p^p(t)\nonumber\\
&=&(y-p^w(s^*_A(x_0)))F_1[x_0]-p^p(s^*_A(x_0))\nonumber
\end{eqnarray}
So $(x_0,s^*_A(x_0))$ lies on $eu_B(y,\cdot)$.

Combined (4.1) and (4.2), we know $eu_B(y,\cdot)$ has common point with $s^*_A(x)$.
So $\forall y>\hat{y}$, there exists $x$ such that $s^*_A(x)\in S^*_B(y)$.
Hence,  $\cup_x Y(x)$ covers $(\hat{y},b_2]$.

(5)If $Y(x)$ is not a unique number.
Then there exists $y_1<y_2\in Y(x)$. For any $y\in(y_1,y_2)$.
$$Y(x_1)<y<Y(x_2)\ \forall x_1<x<x_2$$
Because conclusion in (4), it must be $y\in Y(x)$.
So $Y(x)$ is an interval.

Since there is countable non-overlap interval,
so for almost all $x$, $Y(x)$ contains only one element.
$\Box$

Proof of \textbf{Lemma~\ref{lemma8.5}}.

We consider two cases.

(1)$s^*_A(x)=\lim_{t\rightarrow a_1}s^*_A(t)$

By Lemma~\ref{lemma4}, $(x,s^*_A(x))$ lies on $eu_B(y_0,\cdot)$,
where $y_0=p^w(s^*_A(x))$, $u_B(y_0)=0$.
Furthermore, $\forall y>y_0$, we have $u_B(y)>0$(by just bidding $s^*_A(x)$).
So $\hat{y}=y_0$.

Because $\hat{y}\in Y(x)$ and $Y(x)\geq \hat{y}$, we have
$g(x)=\min Y(x)=\hat{y}$. Then for $y\leq g(x)$, the follower bids zero.
For $y>g(x)$, the follower bids larger than $s^*_A(x)$(o.w. the winning probability is zero,
which leads to zero utilty.)
So the winning probability of the leader with value $x$ is $F_2[g(x)]$.

(2)$s^*_A(x)>\lim_{t\rightarrow a_1}s^*_A(t)$
By Lemma~\ref{lemma5}, if $(x,s^*_A(x))$ does not lie on $eu_B(y,\cdot)$,
then $s^*_A(x)\notin S^*_B[y]$.

(2.1)When $g(x)>\hat{y}$

Since $u_B(g(x))>0$, then by Theorem~\ref{lemma2}, we have
$$S^*_B(y_1)\leq S^*_B(g(x))\leq S^*_B(y_2)\ y_1<g(x)<y_2$$
Then $S^*_B(y_1)<s^*_A(x)\leq S^*_B(y_2)$.
Hence, bidding $s^*_A(x)$, the leader's winning probability is $F_2[g(x)]$.

(2.2)When $g(x)=\hat{y}$

By definition of $\hat{y}$, we have $u_B(y)>0$ for all $y>\hat{y}$.
By Theorem~\ref{lemma2}, we have $g(x)\leq S^*_B(y)$.
For $y>g(x)$, the follower bids larger than $s^*_A(x)$.
For $y\leq g(x)$, the follower bids zero.
So the winning probability of the leader with value $x$ is $F_2[g(x)]$.
$\Box$

Proof of \textbf{Lemma~\ref{lemma12}}.

(1) By Lemma~\ref{lemma8}(3), we know $g(x)$ weakly increases.

(2)Suppose not, then there exists $x_0$ and $d$ such that $\forall x<x_0$, $g(x)<g(x_0)-d$.
By Lemma~\ref{lemma8}(4), pick any $y\in(g(x_0)-d,g(x_0))$, there exists $x_1$ such that $y\in Y(x_1)$.
Then $x_1<x_0$, we also have:
\begin{displaymath}
g(\frac{x_1+x_0}{2})\geq \sup Y(x_1)\geq y>g(x_0)-d
\end{displaymath}
Which contradicts to the supposition.
$\Box$

Proof of \textbf{Theorem~\ref{theorem8}}.

Look at Fig~\ref{fig:2}, $M_2$ is a mapping from $g\in O_2$ to $s\in O_1$, define $s=M_2(g)$ be the solution of
\begin{eqnarray}
\int^{x}_{a_1}f_1(t)g(t)dt=p^w(s_A(x))F_1[x]+p^p(s_A(x))\label{eqa10}
\end{eqnarray}

(1) We prove that $s_A(x)$ does not change after sorting and smoothing, i.e.  $M_2(g)\in O_1$.

Let $a=F_1[x]$, $b=-\int^{x}_{a_1}f_1(t)g(t)dt$, $s_A(x)=t(a(x),b(x))$.
By Lemma~\ref{lemma3.25}, $s_A(x)$ is unique and differentiable.
Since $g$ weakly increases, by Equation~(\ref{eqa10}), $g(x)F_1[x]\geq p^w(s_A(x))F_1[x]$.
Furthermore, $g(x)\geq p^w(s_A(x))$.
\begin{eqnarray}
s'(x)&=&\frac{\partial t}{\partial a}\cdot \frac{\partial a}{\partial x}+\frac{\partial t}{\partial b}\cdot \frac{\partial b}{\partial x}\nonumber\\
&=&\frac{f_1(x)g(x)-p^w(s_A(x))f_1(x)}{(p^w)'(s)F_1[x]+(p^p)'(s)}\geq 0\nonumber
\end{eqnarray}
So s weakly increases.

First, we define $u_B(y,t)$ and $\tilde{u}(y,t)$ as follows:
\begin{eqnarray}
u_B(y,t)&=&[y-p^w(s_A(t))]P_B[t]-p^p(s_A(t))\nonumber\\
\tilde{u}(y,t)&=&[y-p^w(s_A(t))]F_1[t]-p^p(s_A(t))\nonumber
\end{eqnarray}
So $u_B(y,t)$ is the follower's utility with type $y$ and bidding $s_A(t)$,
Since $s$ weakly increases, we have $P_B[t]>F_1[t]$, then $u_B(y,t)\geq \tilde{u}(y,t)$.
Furthermore $\max_tu_B(y,t)\geq \max_t\tilde{u}(y,t)$.
For any $t$, $u_B(y,t)=\tilde{u}(y, F^{-1}_1[P^*_B[s_A(t)]])$, i.e. $\tilde{u}(y,t)$ can achieve any value that $u_B(y,t)$ achieves.
Thus $u_B(y)=\max_tu_B(y,t)=\max_t\tilde{u}(y,t)$, to compute $u_B(y)$, we only need to focus on $\tilde{u}(y,t)$ instead.

Next, we prove that $(x,s_A(x))$ lies on $eu_B(g(x),\cdot)$.
Consider utility of the follower with value $g(x)$.
Since, leader's lowest bid is zero($s_A(a_1)=0$).
so the follower bids between the highest and the lowest of the leader's bid.
Then the derivative $\frac{\partial \tilde{u}}{\partial t}$ should be zero.
\begin{eqnarray}
\frac{\partial \tilde{u}}{\partial t}(g(x),t)&=&[g(x)-p^w(s^*_A(t))]f_1(t)-(p^w)'(s^*_A(t))\cdot(s^*_A)'(t)\cdot F_1[t]\nonumber\\
&&-(p^p)'(s^*_A(t))\cdot(s^*_A)'(t)\nonumber\\
&=&g(x)f_1(t)-g(t)f_1(t)\nonumber
\end{eqnarray}
It's easy to see that $\max_t \tilde{u}(g(x),t)=\tilde{u}(g(x),x)$.
Then $u_B(g(x))=\max_t\tilde{u}(y,t)=\tilde{u}(g(x),x)=[g(x)-p^w(s_A(x))]F_1[t]-p^p(s_A(x))$.
By definition, $(x,s_A(x))$ lies on $eu_B(g(x),\cdot)$.

Let $s^*_A$ defined as before based on $s$. Then $s_A(x)=eu_B(g(x),x)\leq s^*_A(x)\leq s_A(x)$, it must be
$s_A(x)=s^*_A(x)$. That means $s_A(x)$ does not change after sorting and smoothing, so $M_2(s)\in O_1$.

(2) We prove that $M_1\circ M_2=I$.
Suppose otherwise,
then there exists $g$ such that $M_1(M_2(g))=\tilde{g}\neq g$.
Let $s=M_2(g)$.
By Lemma~\ref{lemma9},
we have $\int^{x}_{a_1}f_1(t)\tilde{g}(t)dt=p^w(s^*_A(x))F_1[x]+p^p(s^*_A(x))$
On the other side, by the method used in $M_2$, we have $\int^{x}_{a_1}f_1(t)g(t)dt=p^w(s^*_A(x))F_1[x]+p^p(s^*_A(x))$
So $\int^{x}_{a_1}f_1(t)\tilde{g}(t)dt=\int^{x}_{a_1}f_1(t)g(t)dt$ for any $x$.
Since $g$ and $\tilde{g}$ are left continuous, so
if there $g\neq \tilde{g}$ for some number $x_0$ then $g\neq \tilde{g}$ for some
interval on the left side of $x_0$. Hence, the equation above will not always hold on that interval. Contradiction!
$\Box$

Proof of \textbf{Lemma~\ref{lemmasa1=0}}.

The leader's expected utility is
\begin{eqnarray}
&&\int_{a_1}^{a_2}\{[x-p^w(s^*_A(x))]F_2[g(x)]-p^p(s^*_A(x))\}f_1(x)dx\nonumber\\
&=&\int_{a_1}^{a_2}\{xF_2[g(x)]-p^w(s^*_A(x))F_2[g(x)]-p^p(s^*_A(x))\}f_1(x)dx\nonumber
\end{eqnarray}
When we function $g$ is fixed, $s^*_A(a_1)$ becomes smaller, $s^*_A(x)$ becomes smaller, then we get higher expected utility.
So in the optimal strategy, we must have $s^*_A(a_1)=0$, and $s^*_A(x)$ is the solution of
$\int^{x}_{a_1}f_1(t)g(t)dt=p^w(s^*_A(x))F_1[x]+p^p(s^*_A(x))$
$\Box$

Proof of \textbf{Theorem~\ref{theoremshg}}.

(1)Consider the first case, the other case is similar.
Otherwise we could always increase function $g(x), x\in L$ a little in the first case
and increase the utility. That contradicts to $g$ is optimal.

(2)Create function $\tilde{g}$ that has no image on $(0, b_1)$.
\begin{displaymath}
\tilde{g}(x)=\left\{
\begin{array}{ll}
0 & g(x)\in (0,b_1)\\
g(x) & o.w.
\end{array}
\right.
\end{displaymath}
$\tilde{g}\leq g(x)$, then $s^*_A$ based on $\tilde{g}$ is smaller than $s^*_A$ based on $g$.
While they keep the same winning probability as long as the winning probability $F_2[g]$ is non-zero.
So the expected utility weakly increases when using $\tilde{g}$ instead of $g$.
$\Box$

Proof of \textbf{Lemma~\ref{theorem1st}}.

By Theorem~\ref{theorem7}, function
$$t_0=1\int^{1}_{t_0}\frac{1}{x}dx$$
has a solution $t_0\approx 0.567$.
When $x>t_0$, we have $s^*_A(x)=\frac{1}{F_1[x]}\int_{t_0}^xf_1(t)g(t)dt=1-\frac{t_0}x$
The expected utility is
$$
\int^1_{t_0}(x-1+\frac{t_0}{x}dx\approx0.228
$$
$\Box$

\section{The assumptions are not necessary}
The proof idea is that we can achieve the same largest utility with or without assumptions
and we prove that the optimal strategy in different settings are similar.
\subsection{Assumption~\ref{ass1} is not necessary}

Without Assumption~\ref{ass1}, the follower can give any best response.
We prove that the maximal expected utility of the leader under both Assumptions is same as under only Assumption~\ref{ass2}.

First, the smooth method does not depend on Assumption~\ref{ass1}.
The proof of Theorem~\ref{lemma7} does not depend on Assumption~\ref{ass1}.
From same $s$, we create same $s^*_A$ no matter whether Assumption~\ref{ass1} works.

Second, from the proof of Theorem~\ref{lemma7}, we know there are only countable points $(x,s^*_A(x))$ that lies on multiple $eu$ curves.
For the most value $x$, $(x,s^*_A(x))$ lies on a unique $eu$ curve,
and the winning probability is $F_2[g(x)]$
no matter how the follower chooses between his best responses.
So for same $s^*_A$, the total expected utility of the leader  under both Assumptions is same as under only Assumption~\ref{ass1}.

In conclusion, the optimal strategy of the leader under both Assumptions is same as under only Assumption~\ref{ass2}.
The purpose of Assumption~\ref{ass1} is to define winning probability of all the values, not just on most of the values.

\subsection{Assumption~\ref{ass2} is not necessary}

When considering some other tie-breaking rule rather than always assigning the good to $B$,
our method still works and the optimal strategy is the same.

We restate our situation, without Assumption~\ref{ass1}, the follower may choose any best strategy now.
We would like to prove that for arbitrary tie-breaking rule,
the optimal strategy is same as the optimal strategy under Assumption~\ref{ass2}.

The prove idea is the following:
\begin{enumerate}
\renewcommand{\labelenumi}{(\theenumi)}
\item If tie-breaking rule is assigning the good to $A$, assume $A$'s best strategy is $s$.
\item If tie-breaking rule is assigning the good to $B$, $A$ adopts strategy $s$.
\item If tie-breaking rule is assigning the good to $B$, assume $A$'s best strategy is $\tilde{s}$.
\item If tie-breaking rule is assigning the good to $A$, $A$ adopts strategy $\tilde{s}$.
\end{enumerate}
We will prove that the leader's utility in these four settings has the following order $(4)=(3)\geq (2)=(1)$.
Hence, $s$ is optimal under arbitrary tie-breaking rule is equivalent to $s$ optimal under Assumption~\ref{ass2}.

First we should notice that without Assumption~\ref{ass2}, the follower may not have best response yet.
To ensure the follower choose best response when tie-breaking rule is assign good to $A$, we create $t^+$ with the property that
\begin{eqnarray}
p^w(t^+)&=&p^w(t)\nonumber\\
p^p(t^+)&=&p^p(t)\nonumber\\
P_B(t^+)&=&\lim_{x\rightarrow t^+}P_B(x)\qquad x\textrm{ approaches $t$ from the right side, i.e. } x>t \nonumber\\
x_1<&t^+&<x_2\forall x_1<t<x_2\nonumber\\
t&<&t^+\nonumber
\end{eqnarray}
Now when we say follower's best response, we also include the best bid $t^+$.
The follower can bid $t^+$ to represent a bid that arbitrarily approximate the $t$, but with a weakly higher winning probability than $t$.
Then in $(1)$, $B$ always has best response.
Suppose otherwise $B$ does not have best response with value $y$,
then there is a series of bid $\{t_n\}\rightarrow t$, approaches the largest utility $u_B(y)$.
\begin{eqnarray}
u_B(y)&=&\lim(y-p^w(t_n))P_B[t_n]-p^p(t_n)\nonumber\\
&=&(y-p^w(\lim t_n))\lim P_B[t_n]-p^p(\lim t_n)\nonumber\\
&=&(y-p^w(t))P_B[t]-p^p(t)\textrm{ or }(y-p^w(t^+))P_B[t^+]-p^p(t^+)\nonumber
\end{eqnarray}
So either $t$ or $t^+$ will become the best response.

Now we start to prove the leader's utility in $(2)$ and $(1)$ is equal.
Wlog, we can assume $s$ is sorted in weakly increasing order.

Compared to $(1)$, $B$ has advantage in $(2)$, so $u_B^{(2)}(y)\geq u_B^{(1)}(y)$.
(By $u_B^{(1)}$, we mean the follower's utility function in setting (1), others are similar).
Here, $u_B^{(2)}$ and $u_B^{(1)}$ denote the follower's utility respectively in $(2)$ and $(1)$.
On the other side, the follower could bid $t^+$ in $(2)$ to guarantee the same utility as bidding $t$ in $(1)$.
So we have $u_B^{(2)}(y)=u_B^{(1)}(y)$ and
\begin{eqnarray}
t\in S^{(2)}_B(y)\Rightarrow t^+\in S^{(1)}_B(y)\label{eqa11}
\end{eqnarray}
Here, $S^{(2)}_B(y)$,$S^{(1)}_B(y)$ denote the follower's best response set respectively in $(2)$ and $(1)$.

Now, we prove that the number of the leader's value that the winning probability change in $(1)$ and $(2)$, is countable.
Then the expected the leader's utility between $(1)$ and $(2)$ are same.

Since we introduce $t^+$, we need new version of Theorem~\ref{lemma2}.
\begin{theorem}
$y_1<y_2$, if $\exists a\in S^{(1)}_B(y_1), b\in S^{(1)}_B(y_2)$ but $a> b$, then
either $a=b^+$, $P^{(1)}_B[b^+]=P^{(1)}_B[b]$
or $S^{(1)}_B(y_1)\subseteq S^{(1)}_B(y_2)$, $u^{(1)}_B(y_1)=u^{(1)}_B(y_2)=0$, $P^{(1)}_B[b]=P^{(1)}_B[a]=0$.
\label{newlemma2}
\end{theorem}
We ommitted the proof of this theorem, because the proof is exact the same.

When $y>\hat{y}$(Recall that $\hat{y}=\sup\{y|u_B(y)=0\}$), by this new theorem, $S_B(y)$ is still in order.
If $\hat{y}\leq y_1<y_2$ then $S_B(y_1)\cap S_B(y_2)$ contains at most two elements, something like $t$ and $t^+$.
When $y<\hat{y}$, the follower's best response is smaller than $s_A(x),\forall x>a_1$.
So it does not matter how the follower chooses best response when $y<\hat{y}$, and we let the follower bid zero.
Now we can assume the follower's bids is weakly increasing.

Suppose the leader's winning probability changes with value $x$.
In $(1)$, $A$ wins against $B$ with value $y<y_1$, and loses against $B$ with value $y>y_1$.
In $(2)$, $A$ wins against $B$ with value $y<y_2$, and loses against $B$ with value $y>y_2$.
If we cannot make this supposition, there must exist an interval $L$ and $t$ such that
$\{t,t^+\}\subseteq S_B^{(1)}(y),\ y\in L$,
then $(x,s_A(x))$ must lie on multiple eu lines, which has zero affect to the final expected utility.

Wlog, the supposition is still feasible.
There are two cases to consider: $y_1<y_2$ and $y_1>y_2$.

Case 1:$y_1<y_2$

In (2), we have $\exists t_2\in S_B^{(2)}(y_2-\epsilon)$ such that $t_2<s_A(x)$. Here $\epsilon$ denotes some small enough positive number.
By Equation~(\ref{eqa11}), we have $t_2^+\in S_B^{(1)}(y_2-\epsilon)$.
In (1), we have $t_1\in S^{(1)}_B(y_1+\epsilon)$ such that $t_1\geq s_A(x)$.

So we have $t_2^+<s_A(x)\leq t_1$, while $t_2^+\in S_B^{(1)}(y_2-\epsilon)$ and $t_1\in S^{(1)}_B(y_1+\epsilon)$.
By Theorem~\ref{lemma2}, as long as $\epsilon$ is small enough such that $y_1+\epsilon<y_2-\epsilon$, we get contradiction.
So $y_1$ will be never smaller than $y_2$.

Case 2:$y_1>y_2$

In (2), we have $\exists t_2\in S_B^{(2)}(y_2+\epsilon)$ such that $t_2\geq s_A(x)$. Here $\epsilon$ denotes some small enough positive number.
By Equation~(\ref{eqa11}), we have $s_A(x)\leq t_2^+\in S_B^{(1)}(y_2+\epsilon)$.
In (1), we have $t_1\in S^{(1)}_B(y_1-\epsilon)$ such that $t_1\leq s_A(x)$.

Then we have $y_1<y_2$, $t_2\geq t_1$, $t_2^+\in S_B^{(1)}(y_1-\epsilon)$ and $t_1\in S_B^{(1)}(y_2+\epsilon)$.
Using same argument as Theorem~\ref{lemma2}, we have $P_B^{(1)}[t_2^+]=P_B^{(1)}[t_1]$.

Notice that $t_2\geq s_A(x)\geq t_1$. If $t_2>t_1$, by Theorem~\ref{lemma2}, we get $P_B^{(1)}[t_2^+]=P_B^{(1)}[t_1]=0$.
However, since $s$ is weakly increasing, we have $P_B^{(1)}[t_2^+]\geq P^{(1)}_B[s_A(t)^+]\geq F_1[t]>0$.
Thus $t_2=t_1=s_A(x)$.
We also have $P_B^{(1)}[t_2^+]=P_B^{(1)}[t_2]$, so there is no mass point of $A$'s bid on $t_2$, i.e. $s_A(x)$.
Furthermore, we can conclude among the leader's value, only $A=x$ bids $s_A(x)$.
Then $P_B^{(2)}[t_2]=P_B^{(2)}[t_2]$, $t_2\in S_B^{(2)}(y_1-\epsilon)$.
By assumption, we have $t_2=t_1\in S_B^{(2)}(y_2+\epsilon)$.
Using Theorem~\ref{lemma2}, we can prove that
\begin{displaymath}
S_B^{(2)}(y)=\{t_2\}=\{s_A(x)\}\ \forall y\in(y_2+\epsilon, y_1-\epsilon)
\end{displaymath}
Because the number of interval is countable,
then the number $s_A(x)$ corresponds to the interval is countable,
furthermore the number of $x$ with which the leader has different winning probabilities in $(1)$ and $(2)$ is countable.

$(3)\geq (2)$ is obvious.

At we talk about $(4)=(3)$.
By Lemma~\ref{lemmasa1=0}, we have $s^*_A(a_1)=0$.
If there is no constant interval in $s^*_A$, i.e. $\nexists x$ such that $s^*_A(x)=\lim_{t\rightarrow a_1}s^*_A(t)$,
then there is no tie problem to concern. So it's same optimal strategy whether under Assumption~\ref{ass2} or not.

If there is constant interval in $s^*_A$, i.e. $\hat{x}$ exists.
We consider using the same $s^*_A$ strategy and find the winning probability of the most leader value, does not change.

When $x<\hat{x}$, the leader with value $x$ never wins. In (3) the follower bids zero. In (4), the follower bids $0^+$.
When $x>\hat{x}$, for almost all value $x$, $(x,s^*_A(x))$ lies on a unique eu curve.

When $t>0$, we have $P_B^{(3)}(t)=P_B^{(4)}(t)=P_B^{(4)}(t^+)$, so
$$t>0, t\in S^{(3)}_B(y)\Leftrightarrow t\in S^{(4)}_B(y)$$
Moreover, we have
\begin{eqnarray}
\forall y<g(x) & S^{(3)}_B(y)=S^{(4)}_B(y)&<s^*_A(x)\nonumber\\
\forall y>g(x) & S^{(3)}_B(y)=S^{(4)}_B(y)&>s^*_A(x)\nonumber
\end{eqnarray}
So the winning probability of the leader with value $x$ does not change.
So we have $(4)=(3)$.
\end{appendix}
\end{document}